\documentclass[11pt]{revtex4-2}
\usepackage[a4paper,margin=1in]{geometry}
\usepackage{amsmath, amssymb, amsfonts, latexsym,mathrsfs}
\usepackage{amsthm}
\usepackage{color,graphicx,hyperref,geometry}
\usepackage{inputenc}
\usepackage{footnote}
\usepackage{float}
\usepackage{enumerate}
\usepackage{ulem,booktabs,multirow}
\usepackage{bm,physics,units}
\usepackage{tikz}
\allowdisplaybreaks

\newtheorem{theorem}{Theorem}[section]
\newtheorem{lemma}[theorem]{Lemma}
\newtheorem{proposition}[theorem]{Proposition}

\newtheorem{definition}[theorem]{Definition}

\newtheorem{remark}[theorem]{Remark}
\newcommand{\eps}{\varepsilon}
\newcommand{\Om}{\Omega}
\newcommand{\R}{\mathbb R}
\newcommand{\Qvec}{\mathbf Q}
\newcommand{\C}{\mathbb C}

\newcommand{\e}{\mathrm e}
\newcommand{\loc}{\mathrm{loc}}

\newcommand{\divop}{\operatorname{div}}

\def\P{\mathbf{P}}
\def\I{\mathbf{I}}
\def\n{\mathbf{n}}
\def\x{\mathbf{x}}
\begin{document}
\preprint{APS/123-QED}

\title{Using the Renormalized Energy from Ginzburg-Landau Theory to Compute Solution Landscapes in the Reduced Landau-de Gennes Theory}

\author{Yucen Han}
 \affiliation{Center for Applied Mathematics, Renmin University of China, Beijing 100872, China}

\author{Apala Majumdar}
\email{apala.majumdar@manchester.ac.uk}
\affiliation{Department of Mathematics, University of Manchester, Manchester M13 9PL, United Kingdom}

\begin{abstract}
We study nematic liquid crystal phases confined to bounded, two-dimensional (2D) domains, in the reduced Landau-de Gennes (rLdG) theoretical framework. The rLdG framework is equivalent to the celebrated Ginzburg-Landau (GL) theoretical setting, under certain assumptions. In the asymptotic limit of large domains (or low temperatures), we recall the concept of the finite-dimensional renormalized energy in the GL theory [Bethuel et al., GL Vortices, Prog. Nonlinear Differ. Equ. Appl. 13, Birkhäuser, 1994.] 
and generalize it to polygons with sharp corners. The renormalized energy is fully defined by the triplet, $(\mathbf{a},\mathbf{c}, \mathbf{d}$), of topological degrees of the interior defects ($\mathbf{d}$) and the winding number contribution of polygon vertices ($\mathbf{c}$) and interior defect locations ($\mathbf{a}$), subject to sum rules dictated by the imposed Dirichlet boundary condition. We use the critical points of the renormalized energy to construct initial conditions for the rLdG solver, recover known energy-minimizing and non-energy minimizing saddle points of the rLdG energy on a disc, regular hexagon and rectangle along with the discovery of previously unreported rLdG solution branches. The renormalized energy approach is computationally substantially more efficient than a direct PDE-solver in the rLdG setting but misses solution branches with line defects and finite-size effects.
\end{abstract}

\maketitle

\section{Introduction}
\label{sec:introduction}

Liquid crystals are mesophases that combine fluidity with the ordering characteristics of solids and are classical examples of partially ordered materials with long-range orientational ordering and partial translational ordering \cite{de1993physics}. The nematic phase is the simplest liquid crystal (LC) phase with long-range orientational ordering or special material directions that manifest in direction-dependent optical, physical and mechanical properties.
Nematics have long been the working material of choice for the thriving display industry and modern research has shifted attention to nematic applications in healthcare, sensors, photonics, robotics, anti-counterfeiting and metamaterials, to name a few \cite{lagerwall2012new}. There are of course, many more exotic LC phases such as smectic phases, cholesterics, bent-core phases, twist-bend phases, discotic and columnar phases etc. Mathematical modelling can play a crucial role in not only explaining LC phenomena but also in informing and guiding new interdisciplinary LC research and applications.
The mathematical modelling of LCs is embedded in the study of systems of nonlinear and coupled partial differential equations (PDEs) for the LC order parameters. In this paper, we focus on the fascinating  PDE problems that arise in the context of modelling confined nematic LCs (NLCs) in 2D domains, with special attention to polygons as experimental prototypes.

We work with the celebrated Landau-de Gennes (LdG) theory for modelling nematic phases \cite{de1993physics}. The constituent nematic molecules tend to align with each other, along certain locally preferred or distinguished directions. The LdG theory is a powerful and widely used continuum theory for modelling nematic phases; it describes the nematic phase by a macroscopic order parameter--the LdG $\mathbf{Q}_f$-tensor order parameter which can be interpreted in terms of experimental measurables such as the dielectric anisotropy and birefringence data~\cite{patranabish2021quantum}. The LdG $\mathbf{Q}_f$-order parameter is a symmetric, traceless $3\times 3$ matrix with five degrees of freedom. The eigenvectors model the special material directions and the corresponding eigenvalues measure the degree of orientational order about these directions, respectively~\cite{majumdar2010equilibrium}. The nematic director refers to the eigenvector with the largest positive eigenvalue. In particular, the isotropic or disordered phase is modelled by $\Qvec_f =0$ whereas ordered nematic phases are modelled by $\Qvec_f \neq 0$. A nematic phase is said to be uniaxial if $\Qvec_f$ has two degenerate non-zero eigenvalues and biaxial if the corresponding $\Qvec_f$ has three distinct eigenvalues. 

The LdG theory is a variational theory with an associated free energy, typically of the form
\begin{equation} \label{eq:FLdG}
F_{3D}[\Qvec_f]: = \int_{\Sigma} f_{el}\left(\Qvec_f, \nabla \Qvec_f \right) + f_b\left(\mathbf{Q}_f\right)~dV,
\end{equation} where $\Sigma \subset \mathbb{R}^3$ is a bounded and simply-connected domain with Lipschitz boundary. For the purposes of this paper, we take $\Sigma=\left\{(x,y,z)\in \mathbb{R}^3: (x,y) \in \Omega; z\in [0, h] \right\}$, $\Omega \subset \mathbb{R}^2$ is a bounded and simply-connected domain and $h$ is the domain height. The bulk potential, $f_b$ is a non-convex potential that dictates temperature-driven phase transitions in homogeneous systems; the simplest form of $f_b$ is
\[
f_b\left(\Qvec_f \right) = \frac{A_{3D}}{2}\textrm{tr}\Qvec_f^2 - \frac{B}{3}\textrm{tr}\Qvec_f^3 + \frac{C}{4}\left(\textrm{tr}\Qvec_f^2 \right)^2
\]
where $A_{3D}$ depends linearly on the temperature, $B, C>0$ are material-dependent constants. The minimizers of $f_b$ determine the bulk phase as a function of the temperature and predict a first-order isotropic-nematic phase transition at $A_{3D} = \frac{B^2}{27C}$, i.e., the minimizers of $f_b$ are ordered uniaxial phases for $A_{3D} < \frac{B^2}{27C}$ and are isotropic for $A_{3D}>\frac{B^2}{27C}$. The elastic energy density, $f_{el}$, penalises spatial inhomogeneities. The elastic energy density is usually a quadratic and convex function of $\nabla \Qvec_f$, so that the energy density is lower semi-continuous and coercive in the space of measurable square-integrable $\mathbf{Q}$-tensors with square-integrable first derivatives (i.e., the standard Sobolev space, $W^{1.2}$). A common choice for $f_{el}$ is the one-constant elastic energy density
\[
f_{el} = \frac{L}{2}\left| \nabla \Qvec \right|^2
\]
where $\left| \nabla \Qvec \right|^2 = \sum_{i,j,k=1}^{3} \left(\frac{\partial Q_{ij}}{\partial x_k}\right)^2$ and $L>0$ is an elastic constant. The physically observable NLC equilibria are modelled by global or local minimizers of the LdG free energy, subject to imposed boundary conditions and constraints~\cite{majumdar2010equilibrium}.

In the thin-film limit or in the $h\to 0$ limit for $\Sigma = \Omega \times [0, h]$, it is reasonable to expect that minimizers of the LdG energy are $z$-invariant, the nematic director is planar and the energy-minimizing $\Qvec_f$-tensors have a fixed eigenvector in the $z$-direction with an associated fixed eigenvalue so that $\Qvec_f$ only has two degrees of freedom. This is the so-called reduced Landau-de Gennes (rLdG) model that has been hugely successful for quasi-two dimensional (2D) nematic liquid crystal systems or thin three-dimensional (3D) systems, both for capturing the qualitative properties of physically relevant solutions and for probing into defect cores \cite{rajamanickam2026nematic}. Then the energy-minimizing LdG $\mathbf{Q}_f$-order parameters can be written as:
\begin{equation} \label{eq:Qf}
    \mathbf{Q}_f = \begin{pmatrix}
                    P_{11} + q_3 & P_{12} & 0 \\
                    P_{12} & -P_{11} + q_3 & 0 \\
                    0     & 0              & -2q_3
                    \end{pmatrix}
\end{equation} where $q_3>0$ is some explicitly computable constant in terms of the material constants and prescribed boundary conditions. The rLdG model describes the 2D nematic phase in terms of the rLdG order parameter with simply two degrees of freedom: $P_{11}$ and $P_{12}$ which contain information about the planar nematic director (the leading eigenvector has no out-of-plane component) and the degree of planar ordering. In particular, the zeroes of $P_{11} = P_{12} = 0$ describe defects or vortices in 2D. The corresponding rLdG free energy (obtained by substituting \eqref{eq:Qf} into \eqref{eq:FLdG}, with the one-constant elastic energy density) is given by:
\begin{equation}
    E[\P]: = \int_{\Omega}\frac{L}{2}|\nabla \P|^2+\left(\frac{A}{2}tr\P^2+\frac{C}{4}\left(tr\P^2\right)^2\right) \mathrm{d A},
    \label{eq:p_energy1}
\end{equation} where $A \neq A_{3D}$ above, $L,C$ are positive material-dependent constants. In particular, $\frac{A}{2} = A_{3D}/2-Bq_3+3Cq_3^2$ and $A$ still depends linearly on the temperature, for given positive $B, C, q_3$. The critical points of the rLdG energy are classical solutions of a system of two nonlinear and coupled partial differential equations, which are equivalent to the celebrated Ginzburg-Landau (GL) partial differential equations \cite{bethuel1994ginzburg}. For the special temperature, $A_{3D} = -\frac{B^2}{3C}$, one can show that $q_3 = B/6C$ and there is an exact correspondence between the critical points of the rLdG energy in \eqref{eq:p_energy1} and the LdG energy in \eqref{eq:FLdG}. In general, the critical points of the rLdG energy do not correspond to critical points of the LdG energy and are only good approximations to LdG energy minimizers in the $h\to 0$ limit or the 2D limit.

The solution landscapes of the rLdG energy with the one-constant elastic energy density, i.e., hierarchies of saddle points and minimizers of \eqref{eq:p_energy1} have been extensively studied on regular polygons, irregular polygons and generalized to other models of composite nematic materials, such as ferronematics \cite{rajamanickam2026nematic, han2019transition, han2021solution} etc. The High-Index Optimization-Based Shrinking Dimer
 (HiOSD) method is a popular numerical method for computing critical points of free energies of the type \eqref{eq:p_energy1}, their Morse indices (a measure of the instability or the number of unstable eigendirections of a critical point) and pathways between different critical points reported in \cite{yin2019high, yin2020construction}. For example, the rLdG energy \eqref{eq:p_energy1} has a unique critical point on small 2D domains (or of comparable size to the nematic correlation length) and typically has multiple critical points (multiple minimizers and multiple saddle points) when $\Omega$ is sufficiently large in size. In \cite{yin2020construction}, the authors discuss the rLdG solution landscape on square domains with tangent boundary conditions (the planar director is tangent to the square edges, i.e., aligned along the $x$-axis on the horizontal edges and aligned along the $y$-axis on the vertical edges), as a function of the square edge length $D$, at a fixed $A<0$ in \eqref{eq:p_energy1}. The unique critical point is the WORS (Well Order Reconstruction Solution) with two diagonal defect lines (with $P_{11}=P_{12} = 0$ in \eqref{eq:Qf}) for $D$ small enough. As $D$ increases, the WORS loses stability and multiple solution branches appear: stable solutions without interior defects and multiple unstable saddle points with symmetric constellations of interior point vortices and/or line defects, such that the line defects typically occur along a diagonal or near an edge. In \cite{han2022elastic}, the authors also investigate the effects of elastic anisotropy (when the elastic energy density comprises the Dirichlet energy density, $|\nabla \mathbf{Q}_f|^2$, and additional terms with multiple elastic constants) on rLdG solution landscapes on polygons. In a batch of papers \cite{han2023reduced, shi2023hierarchies, shi2024multistability}, the authors stack the critical points of the rLdG energy to construct critical points of the LdG energy in \eqref{eq:FLdG} on cuboids, illustrating the relevance of rLdG solution landscapes in three-dimensional (3D) scenarios. 

The HiOSD method struggles on large domains due to numerical stiffness and the absence of high-quality initial conditions. To circumvent both difficulties (the curse of high dimensionality and poor initial conditions), we recall the concept of the renormalized energy from GL theory to locate vortices (defects with vanishing order parameter) in energy minimizers, in a certain asymptotic limit relevant for large domains or low temperatures. The mathematical theory of two-dimensional GL vortices is established in the monograph of Bethuel, Brezis, and Hélein \cite{bethuel1994ginzburg}; they study the asymptotic limits of energy minimizers with topologically non-trivial Dirichlet boundary conditions that belong to the zero set of the GL bulk potential. The singular limit is a canonical harmonic map; the canonical harmonic map contains information about the multiplicity of vortices and the profile near a vortex. The optimal locations of the interior vortices (necessitated by the topologically non-trivial boundary datum) are dictated by the minimizers of a renormalized energy, which can be regarded as the finite part of the energy (with the vortex energies removed) \cite{bethuel1994ginzburg}. The renormalized energy is usually the sum of the self-energies of the vortices and the inter-vortex interaction energies (depends on their relative locations), along with a regularizing field. The authors show that the singular limit of GL energy minimizers is a canonical harmonic map with exactly $|d|$ interior vortices of degree either $+1$ or $-1$, subject to a Dirichlet boundary condition of topological degree $d$, where $d$ is an integer. The locations of the interior vortices are minima of the renormalized energy. This framework has been extended by André and Shafrir to minimization problems with more general boundary conditions, including boundary conditions with zeroes \cite{andreshafrir1998}. Beyond minimizers, Lin demonstrates that arbitrary sequences of solutions of the Ginzburg--Landau equations can converge, as the core size tends to zero, to configurations whose vortex locations are critical points of the renormalized energy \cite{lin1995solutions}, thus establishing a direct bridge between infinite-dimensional PDEs and the finite-dimensional renormalized energy. These ideas have subsequently been developed in models motivated by liquid crystals. Alama, Bronsard, and van Brussel study strong and weak tangential anchoring, showing that the singular limit may contain both interior vortices and boundary defects, including half-degree boundary vortices \cite{alamabronsard2023}. Alama, Bronsard, and Colinet establish a corresponding $\Gamma$-convergence theory for tangential boundary conditions and provide a rigorous description of Jacobian concentration and energy lower bounds up to the boundary \cite{alamabronsard2024}. 
Together, these works trace the development of Ginzburg--Landau theory from classical interior-vortex analysis to non-minimizing critical points, boundary defects, and liquid-crystal models, providing the analytical foundation for studying the multiplicity and locations of vortices in geometrically frustrated systems.

In this paper, we use the renormalized energy to study the asymptotic limits of sequences of critical points of the rLdG energy with an energy bound, in terms of a canonical harmonic map. We compute the limiting optimal locations of the interior vortices with $+1$ or $-1$ degree for rLdG energy minimizers on a disc, subject to Dirichlet data of arbitrary topological degree. The optimal locations are the corners of a $N$-polygon of an explicitly computable radius, where $N$ is the number of vortices. We also compute critical points of the renormalized energy for arbitrary combinations of $\left(N, \mathbf{d} \right)$ where $N$ is the number of vortices and $\mathbf{d} = \left(d_1, \ldots, d_N \right)$ are the topological degrees of the vortices, subject to sum rules. These critical points are used to construct high-quality initial conditions for the rLdG solver and we find several interesting saddle points of the rLdG energy on the disc with multiple interior vortices, of different topological degrees.

Further, we extend the renormalized energy to a polygon with straight edges, sharp corners and tangent boundary conditions. The tangent boundary conditions require the planar director to be tangent to the edges, creating discontinuities of director at the vertices or vertex vortices. There are no other boundary vortices but each vertex $b_j$ is characterized by an appropriately defined topological degree, $c_j$. In an unpublished paper, we show that sequences of critical points of the rLdG energy on polygons with tangent boundary conditions (with an energy bound), converge to a canonical harmonic map, fully defined by the triplet $\left(\mathbf{a}, \mathbf{c}, \mathbf{d} \right)$, in the limit of vanishing defect core size. The $\mathbf{a}$ define the locations of the interior vortices and the $\mathbf{d}$ are the associated topological degrees of the interior vortices. The locations of the interior vortices are the critical points of an appropriately defined renormalized energy, for given $\mathbf{d}, \mathbf{c}$. Whilst critical points of the rLdG energy correspond to critical points of the renormalized energy, the reverse relationship need not be true. We use the HiOSD method to efficiently compute the critical points of the renormalized energy for different choices of $\mathbf{d}, \mathbf{c}$, subject to sum rules. These critical points define vortex locations, which define a canonical harmonic map, which can then be used as an initial condition for the rLdG solver. This combined approach of using the renormalized energy in conjunction with the HiOSD recovers the reported rLdG critical points with point defects on a hexagon and a rectangle, but misses critical points with line defects \cite{han2021solution, shi2022nematic}. Importantly, this new approach also discovers new solution branches not reported by means of direct HiOSD searches of the rLdG solution landscape. Notably, this combined approach is significantly faster in terms of computational time, compared to direct PDE-solvers for the rLdG problem.

Our main contribution has been to introduce the renormalized energy as a tool for efficiently canvassing the rLdG solution landscapes, in the limit of vanishing defect core sizes. This work can also be useful for studying defect dynamics in terms of the renormalized energy and could potentially, be extended to study optimal locations of minimal connections or line defects joining pairs of vortices. This paper is organized as follows. We review the rLdG model in \ref{sec:rldg}. In Section~\ref{sec:bbh}, we review the main results from \cite{bethuel1994ginzburg} for smooth, bounded and simply-connected domains with regular boundaries and review the concept of the renormalized energy, followed by some explicit analytic and numerical results for rLdG critical points on a disc. We focus on non-smooth polygons in Section~\ref{sec:polygon}, give the formula for the modified renormalized energy that accounts for interior vortices and vortices at vertices accompanied by results on a hexagon and a rectangle to illustrate computational efficiency in Section~\ref{sec:numericspolygons}. We conclude with some perspectives in Section~\ref{sec:conclusion}.

\section{Reduced Landau-de Gennes model}
\label{sec:rldg}
As described in the Introduction, in the rLdG model, the 2D nematic phase is described in terms of a reduced order parameter - the rLdG order parameter, 
\[ \P \in S_2 = \{\P \in \mathbb R^{2 \times 2}\mid \P = \P^T,  \tr\P = 0\},
\]
i.e., the rLdG order parameter, $\P$, is a symmetric, traceless $2 \times 2$ matrix or a mapping, $\P: G \subset \mathbb{R}^2 \to S_2$ with only two degrees of freedom
\[
\P=
\begin{pmatrix}
P_{11}&P_{12}\\
P_{12}&-P_{11}
\end{pmatrix}.
\] The admissible space for $\P$ is the standard Sobolev space, $W^{1,2}\left(G; S_2 \right)$, subject to Dirichlet boundary conditions as outlined below.
We use the rLdG energy given in \eqref{eq:p_energy1},
where $L$ is a positive elastic constant and we employ the one-constant elastic energy density, 
\[
|\nabla \mathbf{P}|^2 = 2\left( \left(\frac{\partial P_{11}}{\partial x}\right)^2 + \left(\frac{\partial P_{11}}{\partial y}\right)^2 + \left(\frac{\partial P_{12}}{\partial x}\right)^2 + \left(\frac{\partial P_{12}}{\partial y}\right)^2 \right),
\]
$A<0$ is a re-scaled temperature-dependent coefficient, $C$ is a positive  material-dependent constant and $G$ is a bounded and simply connected domain in $\mathbb{R}^2$. Note that we take the coefficient of $\textrm{tr}\P^2$ to be negative so that we work with low temperatures that favour an ordered nematic phase with $\P \neq 0$ in the bulk.
One can also write $\P$ as shown below:
\begin{equation}
    \P = s\left(2 \n\otimes\n- \I_2\right),
    \label{P}
\end{equation}
where $\n = \left(\cos\gamma,\sin\gamma\right)^T$ is the nematic director in the plane, $s$ is a non-negative scalar order parameter that measures the degree of planar order about the nematic director $\n$ and $\I_2$ is the $2\times 2$ identity matrix. Thus we have
\[P_{11} = s\cos\left(2\gamma\right),\ P_{12} = s\sin\left(2\gamma\right).\] The planar defect set is modelled by the nodal set of $s$.

Let $\lambda$ be a characteristic length scale of the 2D domain. We re-scale the energy \eqref{eq:p_energy1} with $\bar{x} = x/\lambda$ so that \eqref{eq:p_energy1} reduces to (up to a constant): 
\begin{equation}
    \bar{E}[\P]: = \int_{\tilde{G}}\frac{1}{2}|\bar{\nabla} \P|^2+\bar{\lambda}^2\left(-\frac{A}{4C}tr\P^2+\frac{1}{8}\left(tr\P^2\right)^2\right) \mathrm{d \bar{A}},
    \label{eq:p_energy}
\end{equation}
where $\tilde{G}$ is the normalized domain and $\bar{\lambda}^2 = \frac{2C\lambda^2}{L}$.
The corresponding Euler-Lagrange equations are
\begin{align}\label{eq:rLdG_EL}
\bar{\Delta} P_{11} = \bar{\lambda}^2\left(P_{11}^2 + P_{12}^2-\frac{A}{2C} \right)P_{11},\\
\bar{\Delta} P_{12} = \bar{\lambda}^2\left(P_{11}^2 + P_{12}^2-\frac{A}{2C} \right)P_{12},
\end{align}
with Dirichlet boundary conditions
\begin{equation}\label{eq:Pg}
P_{11}+P_{12}i = \sqrt{\frac{A}{2C}}g,\ on\ \partial \tilde{G},
\end{equation} where $A, C>0$ by assumption, $i=\sqrt{-1}$ and $g$ is a given smooth complex function on the boundary. In our numerical simulations, we use $A = -\frac{B^2}{2C}$ in \eqref{eq:p_energy1}, which corresponds to $A_{3D} = -\frac{B^2}{3C}$ as used in our previous studies \cite{han2021solution,shi2022nematic}, with MBBA as a representative NLC material and fixed values $B = 0.64 \times 10^4 N m^{-2}$, $C = 0.35 \times 10^4 N m^{-2}$ (see  \cite{wojtowicz1975introduction}).

We can make a direct correspondence between the rLdG energy in \eqref{eq:p_energy} and the celebrated GL theory for superconductivity. Assuming 
\begin{equation}\label{eq:uP}
\sqrt{\frac{A}{2C}}u = P_{11} + P_{12}i = se^{i2\gamma}\in H^1(G;\mathbb{C}),
\end{equation}
and $\bar{x} = x/\lambda$ for domain normalization, the energy in \eqref{eq:p_energy} can be written as a GL energy:
\begin{equation}\label{eq:GL_energy}
E_{\eps}[u] = \int_{\tilde{G}}\frac{1}{2}|\nabla u|^2 + \frac{1}{4\eps^2}(|u|^2-1)^2 d\bar{A},
\end{equation}
where \begin{equation}\label{eq:eps_lambda}
\eps^2 = \frac{2 CL}{\lambda^2 A^2} .
\end{equation}
The corresponding Dirichlet boundary condition for $u$ is
\begin{equation}
u = g\ on\ \partial \tilde{G}.
\end{equation}
Note that, due to the head-to-tail symmetry of nematic directors, the argument of $u$, namely the argument of $P_{11}+P_{12}i$, is $2\gamma$ or twice the director angle.
For notational simplicity, we omit tildes and bars in the remainder of the manuscript.

Consequently, there is a direct correspondence between the rLdG energy functional (and its parameter $\bar{\lambda}$) for the $2 \times 2$ traceless symmetric matrix $P$, and the GL energy functional (and its parameter $\eps$) for the complex scalar field $u$. The limit $\eps \to 0$ corresponds to $\bar{\lambda} \to \infty$. In the subsequent sections, we exploit the limiting behavior of the GL functional as $\eps \to 0$ to deduce the limiting solution $u_*$. As mentioned in the Introduction, there is substantive literature on the characterization of $u_*$ in terms of a finite-dimensional optimization problem or in terms of a suitably-defined finite-dimensional renormalized energy. Via \eqref{eq:uP}, this limiting solution $u_*$ yields a corresponding $\P_*$ and $\P_*$ serves as an excellent initial condition for numerical computations of critical points of the rLdG energy, subject to different kinds of Dirichlet boundary conditions. In general, the non-convex and nonlinear rLdG energy has a complex solution landscape, featured by multiple minima and non energy-minimizing saddle points. For clarity, we provide  the values of $\bar{\lambda}^2$ for the numerically computed critical points of \eqref{eq:p_energy} in the figure captions and the corresponding values of $\eps$ can be readily retrieved using \eqref{eq:eps_lambda}.

\section{Smooth $G$ and $g$}
\label{sec:bbh}
In this section, we review the celebrated results on the characterization of sequences of critical points $\left\{u_{\eps_\ell}\right\}$ of \eqref{eq:GL_energy} as $\eps_\ell \to 0$ (see \cite{bethuel1993asymptotics}). We assume that $g$ is smooth complex-valued function with $|g| = 1$ and 
 the winding number of $g$ is 
\begin{equation}\label{eq:dg}
deg(g,\partial G) = \frac{1}{2\pi}\int_{\partial G} g\times g_{\tau} dS = d,
\end{equation}
where $g_{\tau}$ is the derivative of $g$ in the tangential direction. As mentioned above, these results carry over directly to critical points of the rLdG energy in the $\lambda \to \infty$ limit and the winding number of $g$ on the boundary is therefore twice the winding number of the director orientation. For example, if $G$ is the unit disk and the nematic director is subject to radial/homeotropic boundary conditions, then $d=2$ whereas the winding number of the nematic director is $+1$. 

The admissible set is
\begin{equation}
H_g^1 = \{u\in H^1(G;\mathbb{C});u = g\ on\ \partial G\}.
\end{equation}
\begin{proposition}
The corresponding 
Theorem 0.6 from \cite{bethuel1994ginzburg}. Assume $G$ is a star-shaped, smooth, bounded, simply connected domain in $\mathbb {R} ^{2}$ and $g$ is smooth on $\partial G$. Then there is a subsequence $\eps_\ell\to 0$ and $N$ points $a_1,\ldots,a_N$ in $G$ and a smooth harmonic map $u_*$ from $\bar{G}\backslash\cup_i\{a_i\}$ into $S^1$ with $u_* = g$ on $\partial G$ such that
\[
u_{\eps_\ell}\to u_*\ in\ C_{loc}^k(G\backslash\cup_i\{a_i\})\ \forall k\geq 0\ and\ in\ C^{1,\alpha}(\bar{G}\backslash\cup_i\{a_i\})\ \forall\alpha<1,\] 
where $u_{\eps_\ell}$ is a critical point of the GL energy in \eqref{eq:GL_energy} and a solution of the Euler-Lagrange equations
\begin{align}\label{eq:EL}
-\Delta u &= \frac{1}{\eps^2}u(1-|u|^2)\ in\ G,\\
u&=g\ on\ \partial G.
\end{align}
with parameter $\eps = \eps_\ell$.
Moreover, there exists integers $d_1,\ldots,d_N\in\mathcal{Z}\backslash\{0\}$ satisfying 
\begin{equation}\label{eq:dNg}
\sum_{i = 1}^N d_i = d
\end{equation}
and a smooth harmonic function $\phi:\bar{G}\to\mathcal{R}$ such that
\begin{equation}\label{eq:u*}
u_*(z) = e^{i\phi(z)}\frac{(z-a_1)^{d_1}}{|z-a_1|^{d_1}}\ldots\frac{(z-a_N)^{d_N}}{|z-a_N|^{d_N}}
\end{equation}
where $\phi$ is the solution of the Dirichlet problem
\begin{align}
\Delta\phi &= 0\ in\ G,\\
\phi &= \phi_0\ on\ \partial G.
\end{align}
Here, $\phi_0$ is defined on $\partial G$ by
\begin{equation}\label{eq:phi0}
e^{i\phi_0(z)} = g(z)\frac{|z-a_1|^{d_1}}{(z-a_1)^{d_1}}\ldots\frac{|z-a_N|^{d_N}}{(z-a_N)^{d_N}}
\end{equation}
satisfying $u_*(z) = e^{i\phi_0}\frac{(z-a_1)^{d_1}}{|z-a_1|^{d_1}}\ldots\frac{(z-a_N)^{d_N}}{|z-a_N|^{d_N}} =g(z)$ on the boundary $z\in \partial G$.
In addition, we have 
\begin{equation}
\nabla(u_*(z)\frac{|z-a_j|^{d_j}}{(z-a_j)^{d_j}})(a_j) = 0,\  j = 1,\ldots,N,
\end{equation}
which expresses that $(a_1,\ldots,a_N)$ is a critical point of some appropriate renormalized energy 
\begin{equation}\label{eq:Wa}
W(a_1,\ldots,a_N) = -\pi\sum_{i\neq j}d_id_jln|a_i-a_j| + \frac{1}{2}\int_{\partial G}\Phi_0(g\times g_{\tau})dS - \pi\sum_{i = 1}^N d_iR_0(a_i), 
\end{equation}
where 
\[\Phi_0(z)=\sum_{i = 1}^N d_i ln|z-a_i|+R_0(z),\]
and
\begin{align}
\Delta R_0 &= 0,\ in\ G,\label{eq:Neu1}\\
\frac{\partial R_0}{\partial \nu} &= g\times g_{\tau}-\frac{\partial \sum_{i=1}^N d_i ln|z-a_i|}{\partial \nu}\ on\ \partial G.\label{eq:Neu2}
\end{align}
$\nu$ is the outward normal to $\partial G$ and $\tau$ is a unit tangential vector to $\partial G$ such that $(\nu,\tau)$ is direct, the cross product of 2D vector is controversial, we follow the definition in \cite{cormen2022introduction}, $(x_1,y_1)\times(x_2,y_2) = x_1y_2-x_2y_1$, and the cross product is a scalar, which can be negative or positive.
\end{proposition}

We refer to the points $\left\{a_1, \ldots a_N \right\}$ as the interior point defects of the limiting map, $u_*$ in \eqref{eq:u*} since $u_*$ vanishes at these interior points. In other words, each point $a_i$ is an interior point defect of topological degree $d_i$ such that $\sum_{i = 1}^N d_i = d$ as required, where $d$ is the winding number of the Dirichlet boundary condition $g$.

\subsection{Minimizers of the renormalized energy on a disk with $g\times g_{\tau} = const \ d$}\label{sec:disk}

Let $G$ be the unit disk, $G=\left\{ z \in \mathbb{C}: |z| \leq 1 \right\}.$ We take the corresponding boundary condition to be $g = e^{id\theta}$. Then the winding number of $g$ is $d = \frac{1}{2\pi}\int_{\partial G} g\times g_{\tau} dS$.
\begin{lemma}
  \label{prop:renormalized}
The renormalized energy \eqref{eq:Wa} on a unit disk with $g = e^{id\theta}$ reduces to
\begin{equation}\label{eq:W_disk}
W(a_1,\ldots,a_N) = -\pi\sum_{i\neq j}d_id_j\ln|a_i-a_j| - \pi \sum_{i = 1}^N d_i^2\ln|1 - |a_i|^2|- \pi \sum_{i \neq j}^N d_id_j\ln|1 - a_i\bar{a}_j|.
\end{equation}
\end{lemma}

\begin{proof}

Substituting $\Phi_0$ with $R_0 + \sum_{i = 1}^N d_i ln|r-a_i|$ in \eqref{eq:Wa}, the renormalized energy is
\begin{equation}\label{eq:W}
W(a_1,\ldots,a_N) = -\pi\sum_{i\neq j}d_id_jln|a_i-a_j| +  \frac{d}{2}\sum_{i=1}^N\int_{\partial G}d_i ln|z-a_i|dS + \pi (dR_0(0)- \sum_{i = 1}^N d_iR_0(a_i)), 
\end{equation}
where each line integral (around the unit circle $\partial G$) vanishes from Jensen's inequality ($|a_i| < 1$) and
\begin{align}
\Delta R_0 &= 0,\ in\ G\\
\frac{\partial R_0}{\partial \nu} &= d-\frac{\partial \sum_{i=1}^N d_i ln|z-a_i|}{\partial \nu},\ on\ \partial G.
\end{align}
Recall that for $z \in \partial G$, the following identity holds 
\begin{equation}
\frac{\partial}{\partial \nu} \left( \ln|1 - z \bar{a}_j| + \ln|z - a_j| \right) = 1 \quad \text{on } \partial G.
\end{equation}
The proof of this identity is elementary and is provided in Appendix B for completeness.
Subsequently, since $\ln|1-z\bar{a}_i|$ is analytic in $G$, the analytic solution of $R_0$ is 
\begin{equation}
R_0(z) = \sum_{j=1}^N d_j \ln|1 - z\bar{a}_j| +C
\end{equation}
We obtain the formula of renormalized energy $W$ on a unit disk \eqref{eq:W_disk} by substituting the above equation into \eqref{eq:W}.
\end{proof}

\begin{remark}\label{rk:rotation_invariance} (Rotation Invariance): 
One can show that $W(\mathbf{R}a_1,\ldots,\mathbf{R}a_N) = W(a_1,\ldots,a_N)$, for any rotation matrix $\mathbf{R}$. Hence, if $\left\{a_1,\ldots,a_N \right\}$ is a critical point of the renormalized energy on a unit disk, then so are the rotated sets $\mathbf{R}a_1,\ldots,\mathbf{R}a_N$. In the following, we assume that $a_1$ is real-valued and calculate the relative location of $N$ interior defect points.
\end{remark}

\begin{remark}\label{rk:opposite_sign} (Opposite Sign):
$W(a_1,\ldots,a_N)$ with $d_1,\ldots,d_N$ and $d$ is equivalent to $W(a_1,\ldots,a_N)$ with $-d_1,\ldots,-d_N$ and $-d$. We give an example with $d = +1$ and $d = -1$ in Fig. \ref{fig:disk_+1_-1}. In other numerical results, we only focus on the case with $d>0$. 
\end{remark}

\begin{theorem}[Minimizers of the renormalized energy on a unit disk]\label{th:disk}
Let \(N\ge 2\), and suppose either \(d_1=\ldots=d_N=1\) or
\(d_1=\ldots=d_N=-1\). Let
\(a_1,\ldots,a_N\in\mathbb D:=\{z\in\mathbb C:|z|<1\}\). Then the renormalized energy on a unit disk is
\begin{equation}\label{eq:W_N_polygon}
W(a_1,\ldots,a_N)
=
-\pi\sum_{\substack{i,j=1\\ i\ne j}}^N \ln |a_i-a_j|
-\pi\sum_{i=1}^N \ln(1-|a_i|^2)
-\pi\sum_{\substack{i,j=1\\ i\ne j}}^N \ln |1-a_i\overline{a_j}|.
\end{equation}
\(W\) attains its minimum. Moreover, up to a permutation of the points
and a rotation, the unique minimizer has $N$-polygon configuration with the locations of defects at 
\begin{equation}\label{eq:Npolygon}
\{a_1,\ldots,a_N\}
=
\left\{r_N e^{i(\theta+2\pi k/N)}:k=0,\ldots,N-1\right\},
\qquad \theta\in\mathbb R,
\end{equation}
where
\[
r_N=
\left(\frac{N-1}{3N-1}\right)^{1/(2N)}.
\]
The minimum value of $W$ in \eqref{eq:W_N_polygon} is
\begin{equation}\label{eq:Wmin}
W_{\min}
=
-\pi N\ln\left[Nr_N^{N-1}(1-r_N^{2N})\right].
\end{equation}
\end{theorem}

\begin{proof}
Set
\[
\Phi_N(a_1,\ldots,a_N):=\exp(-W(a_1,\ldots,a_N)/\pi).
\]
Then minimizing $W$ is equivalent to maximizing $\Phi_N$. From the definition of
\(W\),
\[
\Phi_N
=
\prod_{i=1}^N(1-|a_i|^2)
\prod_{\substack{i,j=1\\ i\ne j}}^N |a_i-a_j|
\prod_{\substack{i,j=1\\ i\ne j}}^N |1-a_i\overline{a_j}|.
\]
Due to the ordering, equivalently,
\begin{equation}\label{eq:Phi-product-revised}
\Phi_N
=
\prod_{1\le i<j\le N}|a_i-a_j|^2
\prod_{i,j=1}^N |1-a_i\overline{a_j}|.
\end{equation}
The function \(\Phi_N\) extends continuously to \(\overline{\mathbb D}^{\,N}\). It
vanishes whenever two points coincide or at least one point belongs to
\(\partial\mathbb D\). Since \(\Phi_N>0\) for any configuration of distinct points in
\(\mathbb D\), a maximizer of \(\Phi_N\) exists, lies in \(\mathbb D^N\), i.e. $|a_i|<1$, $i = 1,\ldots,N$, and consists
of pairwise distinct points, i.e., $a_i\neq a_j$, for $i\neq j$ and $i,j = 1,\ldots,N$.

Let
\[
t_i:=|a_i|^2,
\qquad
S_N(x):=\sum_{m=0}^{N-1}x^m.
\]
We first estimate the Vandermonde factor $\prod_{1\le i<j\le N}|a_i-a_j|^2$. For any \(\tau>0\), consider
\[
\mathbf{M}=\left(\left(\frac{a_i}{\sqrt\tau}\right)^m\right)_{
\substack{1\le i\le N\\0\le m\le N-1}}.
\]
The Vandermonde determinant formula gives
\begin{equation}\label{eq:MV}
\det \mathbf{M}
=
\tau^{-N(N-1)/4}
\prod_{1\le i<j\le N}(a_j-a_i).
\end{equation}
By Hadamard's determinant inequality,
\begin{equation}\label{eq:Hadamard}
|\det M|^2
\le
\prod_{i=1}^N\sum_{m=0}^{N-1}\left|\frac{a_i}{\sqrt\tau}\right|^{2m}
=
\prod_{i=1}^N S_N\left(\frac{t_i}{\tau}\right).
\end{equation}
Subsequently, due to \eqref{eq:MV} and \eqref{eq:Hadamard}
\begin{equation}\label{eq:Vandermonde-Hadamard-bound-revised}
\prod_{1\le i<j\le N}|a_i-a_j|^2
\le
\tau^{N(N-1)/2}
\prod_{i=1}^N S_N\left(\frac{t_i}{\tau}\right).
\end{equation}

Next we prove
\begin{equation}\label{eq:boundary-product-bound-revised}
\prod_{i,j=1}^N |1-a_i\overline{a_j}|
\le
\prod_{i=1}^N(1-t_i^N).
\end{equation}
Define
\[
A_{ij}:=S_N(a_i\overline{a_j}),
\qquad
B_{ij}:=\frac{1}{1-(a_i\overline{a_j})^N},
\qquad
C_{ij}:=\frac{1}{1-a_i\overline{a_j}}.
\]
Then \(\mathbf{C}=\mathbf{A}\circ \mathbf{B}\), where \(\circ\) denotes the Hadamard product, which is element-wise product.
 Since
\[
A_{ij}
= S_N(a_i\overline{a_j}) = 
\sum_{m=0}^{N-1}a_i^m\overline{a_j}^{\,m}
=
\langle v_{0i},v_{0j}\rangle,
\qquad
v_{0i}=(1,a_i,\ldots,a_i^{N-1}),
\]
and applying the Vandermonde determinant formula again, we have 
\begin{equation}\label{eq:detA-revised}
\det \mathbf{A} = |\det \mathbf{V}_0|^2 = \prod_{1\le i<j\le N}|a_i-a_j|^2>0
\end{equation}
i.e. the matrix \(\mathbf{A}\) is positive definite.
 The matrix \(\mathbf{B}\) is positive semidefinite, because
\[
B_{ij}
=
\sum_{k=0}^{\infty}(a_i^N)^k\overline{(a_j^N)^k}
=
\langle w_{0i},w_{0j} \rangle,
\qquad
w_{0i}=(1,a_i^N,a_i^{2N},\ldots).
\]
By the Cauchy determinant formula \cite[Eq.~(4)]{Schechter1959},
\begin{equation}\label{eq:detC-revised}
\det \mathbf{C}
=
\frac{
\prod_{1\le i<j\le N}(a_i-a_j)
\prod_{1\le i<j\le N}(\overline{a_i}-\overline{a_j})
}{
\prod_{i,j=1}^N(1-a_i\overline{a_j})
}
=
\frac{
\prod_{1\le i<j\le N}|a_i-a_j|^2
}{
\prod_{i,j=1}^N(1-a_i\overline{a_j})
}.
\end{equation}
We claim
\begin{equation}\label{eq:aiaj}
\prod_{i,j=1}^N(1-a_i\overline{a_j})
=
\prod_{i=1}^N(1-t_i)
\prod_{1\le i<j\le N}|1-a_i\overline{a_j}|^2
=
\prod_{i,j=1}^N|1-a_i\overline{a_j}|>0.
\end{equation}
Since \(t_i=|a_i|^2\), the diagonal factors are
\[
1-a_i\overline{a_i}=1-t_i = |1-a_i\overline{a_i}|>0.
\]
For \(i\ne j\), the factors with indices \((i,j)\) and \((j,i)\) are conjugate
\[
1-a_j\overline{a_i}
=
\overline{1-a_i\overline{a_j}}.
\]
Hence
\[
(1-a_i\overline{a_j})(1-a_j\overline{a_i})
=
|1-a_i\overline{a_j}|^2 = |1-a_i\overline{a_j}||1-a_j\overline{a_i}|.
\]
Therefore we have \eqref{eq:aiaj}.
Oppenheim's determinant inequality \cite{Oppenheim1930,HornJohnson2012} for Hadamard products gives
\begin{equation}\label{eq:Oppenheim-step-revised}
\det \mathbf{C}
=\det(\mathbf{A}\circ \mathbf{B})
\ge
\det \mathbf{A}\prod_{i=1}^N B_{ii}
=
\det \mathbf{A}\prod_{i=1}^N\frac{1}{1-t_i^N}.
\end{equation}
Substituting \eqref{eq:detA-revised} and \eqref{eq:detC-revised} into
\eqref{eq:Oppenheim-step-revised} yields \eqref{eq:boundary-product-bound-revised}.
Combining \eqref{eq:Phi-product-revised},
\eqref{eq:Vandermonde-Hadamard-bound-revised}, and
\eqref{eq:boundary-product-bound-revised}, we obtain, for every \(\tau>0\),
\begin{equation}\label{eq:combined-product-bound-revised}
\Phi_N
\le
\tau^{N(N-1)/2}
\prod_{i=1}^N
\left[(1-t_i^N)S_N\left(\frac{t_i}{\tau}\right)\right].
\end{equation}

We choose the parameter \(\tau\) as
\[
\tau^N=\frac{N-1}{3N-1}.
\]
Thus
\[
\sqrt{\tau}
=
\left(\tau^N\right)^{1/(2N)}
=
\left(\frac{N-1}{3N-1}\right)^{1/(2N)}
=r_N .
\] We claim that, for every \(0\le t<1\),
\begin{equation}\label{eq:one-variable-bound-revised}
(1-t^N)S_N\left(\frac{t}{\tau}\right)
\le
N(1-\tau^N),
\end{equation}
the equality holds if and only if \(t=\tau\).
Indeed, we set \(x=t/\tau\) where $\tau^N=\frac{N-1}{3N-1}$. Since $a_i$, $i = 1,\ldots,N$ are inside $\mathbb{D}$, we assume \(t<1\), and subsequently we have
\(0\le x<\tau^{-1}\). The inequality in
\eqref{eq:one-variable-bound-revised} becomes
\[
F_N(x):=(1-\tau^Nx^N)\sum_{m=0}^{N-1}x^m\le N(1-\tau^N).
\]
The derivative of $F_N$ is
\begin{equation}\label{eq:FN1}
F_N'(x)
=
\sum_{k=0}^{N-2}(k+1)x^k
-\tau^N\sum_{k=N-1}^{2N-2}(k+1)x^k.
\end{equation}
Since
\[
F_N'(1)
=\frac{N(N-1)}2
-\tau^N\sum_{m=N}^{2N-1}m
=\frac{N(N-1)}2
-\frac{N-1}{3N-1}\cdot\frac{N(3N-1)}2
=0,
\]
we can write $F_N'(x)$ as
\begin{equation}\label{eq:FN2}
F_N'(x)=(1-x)P_N(x),
\end{equation}
where \(P_N(x)=\sum_{k=0}^{2N-3}p_kx^k\) is a polynomial of $x$. The comparison of coefficients between \eqref{eq:FN1} and \eqref{eq:FN2} gives
\[
p_k=
\begin{cases}
\displaystyle \sum_{m=1}^{k+1}m
=\dfrac{(k+1)(k+2)}2,
&0\le k\le N-2,\\[1.2em]
\displaystyle \sum_{m=1}^{N-1}m-\tau^N\sum_{m=N}^{k+1}m
=\frac{(N-1)(2N-1)}{3N-1},
&N-1\le k\le 2N-3.
\end{cases}
\]
The first group of coefficients with $0\leq k\leq N-2$ and the second group with $N-1\leq k\leq 2N-3$ are both positive.
Hence every coefficient of \(P_N\) is strictly positive. Therefore
\(P_N(x)>0\) for \(x\ge0\), and so \(F_N'(x)>0\) on \([0,1)\) and
\(F_N'(x)<0\) on \((1,\tau^{-1})\). Thus \(F_N\) has its unique maximum at
\(x=1\), where
\[
F_N(1)=N(1-\tau^N).
\]
This proves \eqref{eq:one-variable-bound-revised}.
Applying \eqref{eq:one-variable-bound-revised} to each factor in
\eqref{eq:combined-product-bound-revised}, we obtain
\[
\Phi_N
\le
\tau^{N(N-1)/2}\,[N(1-\tau^N)]^N
=
\left[N\tau^{(N-1)/2}(1-\tau^N)\right]^N.
\]
The upper bound of $\Phi_N$ is
\begin{equation}\label{eq:sharp-Phi-bound-revised}
\Phi_N
\le
\left[Nr_N^{N-1}(1-r_N^{2N})\right]^N
\end{equation} where $r_N = \left(\frac{N-1}{3N-1}\right)^{1/(2N)}$.

It remains to discuss equality. Let
\[
a_k=re^{i(\theta+2\pi k/N)},
\qquad k=0,\ldots,N-1.
\]
For this regular polygon with $N$ vertices,
\[
\prod_{\substack{j=0\\j\ne k}}^{N-1}|a_k-a_j|=Nr^{N-1},
\qquad
\prod_{j=0}^{N-1}|1-a_k\overline{a_j}|=1-r^{2N}.
\]
Therefore
\[
\Phi_N
=
\left[Nr^{N-1}(1-r^{2N})\right]^N.
\]
Taking \(r=r_N\) gives equality in \eqref{eq:sharp-Phi-bound-revised}; hence the
upper bound is sharp.

Conversely, let \((a_1,\ldots,a_N)\) be a maximizer, which attains the bound \eqref{eq:sharp-Phi-bound-revised}, then every
inequality in the chain must be an equality for any maximizer.
Equality in \eqref{eq:one-variable-bound-revised} gives
\[
|a_i|^2=\tau,
\qquad i=1,\ldots,N.
\]
Thus all points lie on the circle of radius \(r_N=\sqrt\tau\) where $\tau^N=\frac{N-1}{3N-1}.
$
Equality in
Hadamard's determinant inequality in
\eqref{eq:Vandermonde-Hadamard-bound-revised} implies that the rows of
\[
\left(\left(\frac{a_i}{r_N}\right)^m\right)_{
\substack{1\le i\le N\\0\le m\le N-1}}
\]
are pairwise orthogonal. Put \(u_i=a_i/r_N\). Then \(|u_i|=1\), and for
\(i\ne j\),
\[
\sum_{m=0}^{N-1}(u_i\overline{u_j})^m=0.
\]
Since $\bar{u}_j = u_j^{-1}$, we can rewrite the summation as
\begin{equation}
\sum_{m=0}^{N-1} (u_i \bar{u}_j)^m = \sum_{m=0}^{N-1} z^m = \frac{1 - z^N}{1 - z} = 0,
\end{equation}
where  $z = u_i u_j^{-1}$.
Since $u_i \neq u_j$ for $i \neq j$, we have $z \neq 1$, which implies $z^N = 1$. Thus, the ratio of any two distinct elements must be a non-trivial $N$-th root of unity. 
Thus, \[
\{u_1,\ldots,u_N\}
=
\{e^{i\theta}e^{2\pi i k/N}:k=0,\ldots,N-1\}.
\]
Subsequently, we obtain the location of interior defects in \eqref{eq:Npolygon} and the minimum of $W$ in \eqref{eq:Wmin}.
\end{proof}

\begin{remark} According to Theorem IX.1 in \cite{bethuel1994ginzburg}, if $u_{\eps}$ is a global minimizer of \eqref{eq:GL_energy}, then $u_{\eps}$ has exactly $|d|$ zeros in $G$ and each zero is of degree $d_i=+1$, $i = 1,\ldots, N$, if $d>0$ ($d_i=-1$, $i = 1,\ldots, N$, if $d<0$). Without loss of generality, we assume $d>0$ in the following.
According to Theorem IX.3 in \cite{bethuel1994ginzburg}, we have
\begin{equation}
E_{\eps}(u_{\eps}) = N\pi|log\eps| + N\gamma + min_{\mathbb{D}^N}W + o(1),
\end{equation}
where $\gamma$ is some universal constant.
In Theorem \ref{th:disk}, we prove that for a unit disk, the interior defect set is defined by the unique minimizer of the renormalized energy in \eqref{eq:Npolygon}
up to a permutation of the points
and a rotation. Hence, as $\eps\to0$, the GL energy minimizer on a unit disk, labelled by $u_\eps$, has a discrete set of $d$ interior defects, $a_i$, each of degree $d_i =+1$. The $a_i$'s are pinned to the vertices of a regular $N$-polygon inside the unit disk, with radius $r_N=
\left(\frac{N-1}{3N-1}\right)^{1/(2N)}$ where $N=|d|$. For $N=2$, the interior defects are diametrically opposite each other, with $r_2 =\frac{1}{\sqrt[4]{5}} \approx 0.67$. Notably, the radius approaches unity as $N \to \infty$. In other words, the interior defects of energy minimizers are equipartitioned so that the inter-defect distance is constant and are expelled towards the boundary as $d$ gets very large. See numerical results in Fig. \ref{fig:pl_disk}.
\end{remark}

\begin{remark} In Theorem~\ref{th:disk}, we show that $W$ attains a unique minimum when the interior defects are located at the vertices of a $N$-polygon, inside the unit disk when the $d_i=\pm 1$. If the constraint $d_i=\pm 1$ is relaxed, then $W$ may admit other minima with other constellations of interior defects, $\left\{a_i \right\}$, with $d_i \neq \pm 1$ for all $i$. We have not derived rigorous results for the locations of $\left\{a_i \right\}$ with arbitrary combinations of $d_i$'s subject to the topological constraint, $\sum d_i = d$.
\end{remark}

\begin{figure}
    \centering
\includegraphics[width=0.6\textwidth]{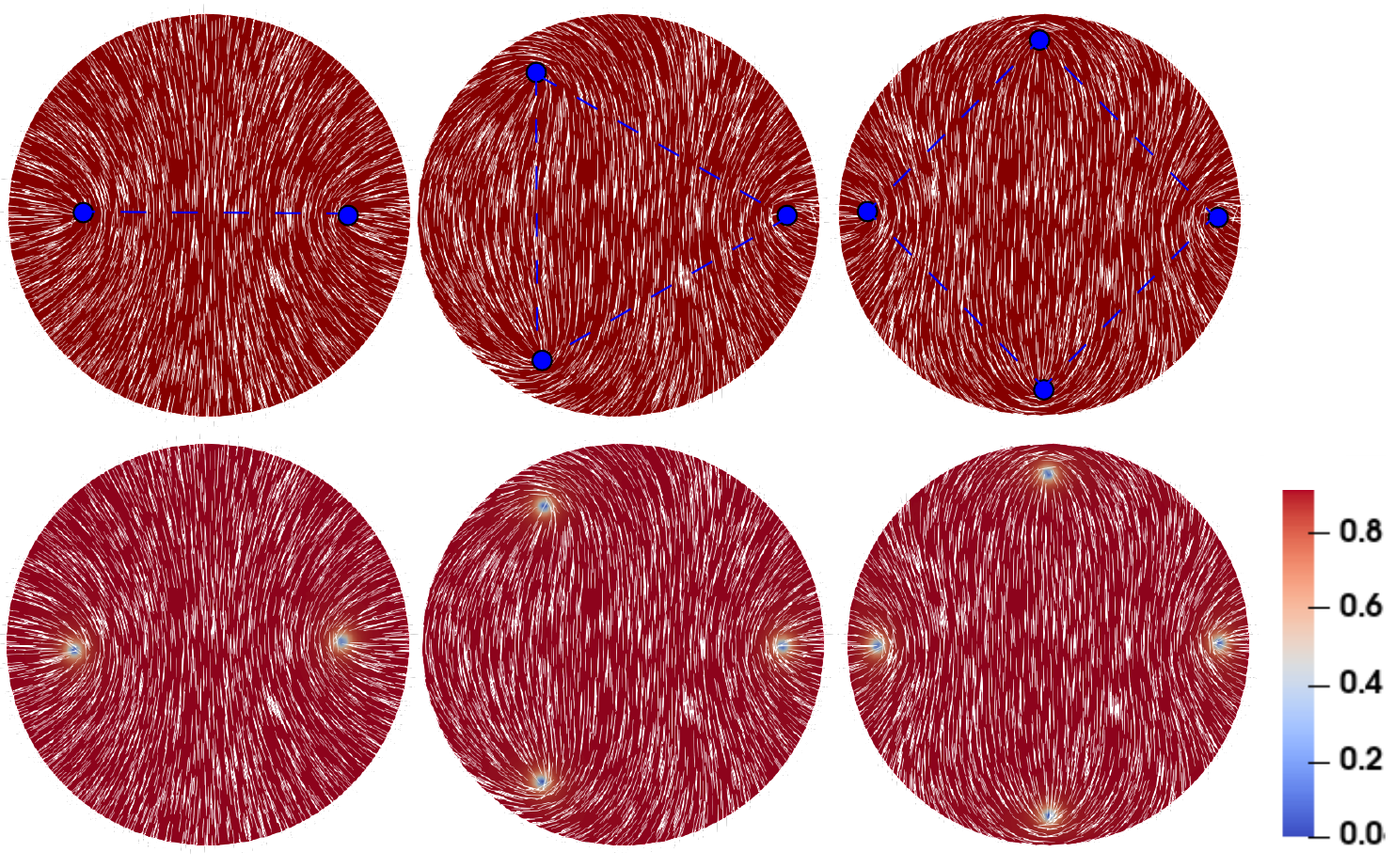}
\caption{\textbf{First row:} the limiting $N$-polygon configurations $\P_*$ related to $u_*$ by the relations \eqref{eq:uP}. The white lines follow the planar nematic director $\mathbf{n}=(\cos(angle(P_{*11},P_{*12})/2),\sin(angle(P_{*11},P_{*12})/2))$. The limiting solution $u_*$ is in \eqref{eq:u*} with the interior defects, $a_i$, determined by the $N$-polygon arrangement in Theorem \ref{th:disk}. The blue dots indicate the location of $+1/2$ defects, the dashed line indicates symmetry of defect locations, and the confined region (here a disk) is rendered with a red background as the order parameter is constant. 
\textbf{Second row:} the corresponding rLdG solution $\P$ with $\lambda^2 = 2000$, by using $\P_*$ as an initial condition. The symmetry is maintained, i.e. there are $N$ $+1/2$-defects placed on the diagonal of disk ($N = 2$) or the corners of regular polygon with $N$ edges ($N=3,4$).
The color encodes the 2D nematic order, $s = \sqrt{P_{11}^2+P_{12}^2}$, and blue represents low nematic order manifested as defects. The white lines follow the planar nematic director $\mathbf{n}=(\cos(angle(P_{11},P_{12})/2),\sin(angle(P_{11},P_{12})/2))$.
All subsequent figures use the same representation and color bar for rLdG solutions.}
\label{fig:pl_disk}
\end{figure}
\subsection{Numerical Method}
\label{sec:numerics}
We use SaddleScape V1.0 \cite{liu2026saddlescapev10pythonpackage}, a python package for constructing solution landscapes via the High-Index Saddle Dynamics (HiSD) algorithm, equipped with Nesterov's accelerated gradient method, to locate the minimum and saddle points of the renormalized energy $W$ \cite{liu2026saddlescapev10pythonpackage}. For a system with $N$ interior defects, the state variables are the complex coordinates of the defects $a_i = x_i + i y_i$. To eliminate continuous rotational degeneracy in the disk domain, the first defect is constrained to the real axis ($y_1 = 0$), reducing the state space dimension to $2N - 1$.

For a given configuration of $N$ defects, we generate $20^N$ random initial configurations to ensure a thorough exploration of the landscape. The position of the first defect $x_1$ is uniformly sampled on the interval $[-1, 1]$, while the remaining $N-1$ defects are uniformly sampled within the unit disk using $r \sim \sqrt{\mathcal{U}(0,1)}$ and $\theta \sim \mathcal{U}(0, 2\pi)$. Rather than fixing a target Morse index a priori, the algorithm dynamically determines the target index of HiSD from $k$ to $1$ for each initial point by calculating $k$, the number of negative eigenvalues of the Hessian matrix at the initial location. 
The HiSD iterations proceed with a fixed time step of $\Delta t = 10^{-3}$ and a maximum of $1000$ iterations. The search is deemed converged if the gradient norm drops below a tolerance of $10^{-5}$. If the gradient norm $||\nabla_{\mathbf{a}} W||>10$ at any point during the iterations, the current trajectory is classified as divergent and the solver skips to the next initial point.

To systematically discover all admissible critical points of $W$ for a given boundary condition, with total topological degree $d$, we parallelize the search over different combinations of individual defect degrees $d_i$. The individual degrees are restricted to $d_i \in \{1, 2, -1, -2\}$. We iterate through all possible combinations for $1 \le N \le N_{max}$, filtering those that satisfy the topological constraint $\sum_{i=1}^N d_i = d$. The independent configuration tasks are then distributed across a multi-core CPU process pool for parallel execution.

Once the critical points of $W$ are identified, we are able to reconstruct the limiting solution $u_*$ according to Eq. \eqref{eq:u*} and obtain $\P_*$ via \eqref{eq:uP} by using the optimized defect locations $a_i$ and their corresponding degrees $d_i$.
Then, $\P_*$ serves as a high-quality initial condition 
for the rLdG computation with Dirichlet boundary conditions in \eqref{eq:Pg}.
We solve the corresponding Euler-Lagrange equation in \eqref{eq:rLdG_EL} by the finite element method. 
We use the FEniCS package \cite{Logg2012FEniCS}, with a first-order Lagrange element, mesh size less than $1/64$, Newton nonlinear variational solver, an LU direct linear solver with the absolute tolerance $10^{-5}$ and the relative tolerance $10^{-3}$ along with a maximum limit of $100$ iterations.

\subsection{Numerical Results}
\label{sec:discresults}
\subsubsection{Homeotropic boundary conditions on the unit disk}
We consider the simplest case of homeotropic boundary conditions first, for which $g\times g_{\tau}\equiv 2$. Using the relations in \eqref{eq:uP},  homeotropic boundary conditions require the director to be radial or normal to the boundary. 
The condition $g\times g_{\tau}\equiv 2$ also corresponds to
tangential boundary conditions which imply that the corresponding nematic director follows the unit circle or is tangential to the boundary. The limiting profiles and the corresponding rLdG profiles for tangential and homeotropic boundary conditions have the same defect locations but the corresponding directors are mutually orthogonal to each other. 

In what follows, we compute critical points of the renormalized energy, $W$ in \eqref{eq:W_disk}, for different combinations of $d_i$ and N. The Morse index of a critical point is a measure of its instability or the number of its unstable directions. In particular, a local minimum is an index-$0$ critical point. We label the Morse index of a critical point of $W$ by $Windex$. The critical points of $W$ are used to construct limiting solutions $u_*$ as above, which then yield high-quality initial conditions $\P_*$ for the rLdG solver. The Morse index of the corresponding critical point of the rLdG energy in \eqref{eq:p_energy} (when numerically computed) is referred to as $index$. It is not guaranteed that every critical point of $W$ corresponds to a critical point of the rLdG energy in \eqref{eq:p_energy}.  
 With $N_{max}=4$ and restricting $|d_i|\leq 2$, we explore the following $7$ valid topological combinations of $d_i$:
\begin{itemize}
    \item $N=1$: $\mathbf{d} = (2)$.
    \item $N=2$: $\mathbf{d} =(1, 1)$.
    \item $N=3$: $\mathbf{d} = (2, 1, -1)$, $\mathbf{d} = (2, 2, -2)$.
    \item $N=4$: $\mathbf{d} = (1, 1, 1, -1)$, $\mathbf{d} = (2, 1, 1, -2)$, $\mathbf{d} = (2, 2, -1, -1)$.
\end{itemize}
We find three critical points of $W$ in total. For $N = 1$ and $\mathbf{d} = (2)$, there is one minimizer of $W$ corresponding to the radial solution with a $+2$ defect at the center of the disk. This is analytically shown to be the unique minimizer of the corresponding renormalized energy in the proof of Theorem VIII.6 in \cite{bethuel1994ginzburg}. In Fig.~\ref{fig:disk_2}(a), we plot the corresponding limiting rLdG profile, $\P_*$ constructed from the limiting solution $u_*$. The central defect in the rLdG model has degree $+1$, since the GL degree is twice the topological degree of the rLdG counterparts. Similar remarks apply to the remainder of this discussion.  For $N =2$ and $\mathbf{d} = (1, 1)$, we find two critical points of $W$. The minimizer of $W$ is the polar solution derived in Theorem~\ref{th:disk} with two $+1/2$ defects at symmetric positions on a diameter of the disk. The limiting solution $\P_*$ and corresponding rLdG solution are shown in Fig. \ref{fig:disk_2}(b). The second critical point is an index-$1$ critical point; the corresponding limiting solution $\P_*$ has a $+1/2$ defect at the disk center and another $+1/2$ defect away from the center (Fig. \ref{fig:disk_2}(c)). We do not find a rLdG solution corresponding to this limiting solution. We do not obtain any other critical points for other choices of $N$ and $\mathbf{d}$. 

\begin{figure}[ht]
    \centering
\includegraphics[width=\textwidth]{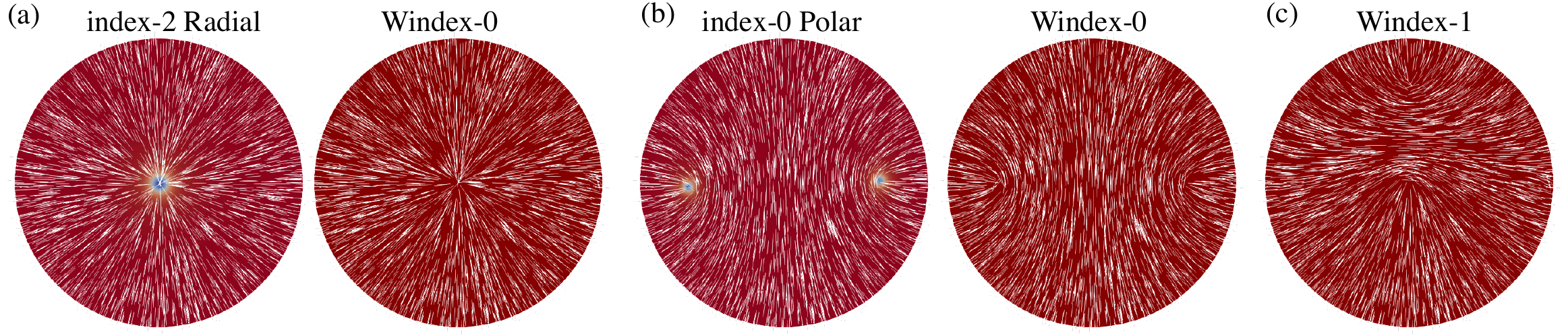}
\caption{(a) Radial solution: The limiting solution $\P_*$ corresponding to a minimizer of $W$ with $N=1$, $d=\left\{2 \right\}$ (right) and corresponding rLdG solution with $\lambda^2 = 2000$ (left). (b) Polar solution: The limiting solution $\P_*$ corresponding to a minimizer of $W$ with $N=2$, $d=\left\{1,1 \right\}$ (right) and corresponding rLdG solution with $\lambda^2 = 2000$ (left). (c) The limiting solution $\P_*$ corresponding to an index-$1$ critical point of $W$ with $N=2$, $d=\left\{1,1 \right\}$. The profiles of the limiting solutions, $\P_*$, constructed from the critical points of the renormalized energy $W$ in \eqref{eq:W_disk}.}
\label{fig:disk_2}
\end{figure}

\subsubsection{Other boundary conditions}

According to Remark \ref{rk:opposite_sign},  the renormalized energy 
$W$ has the same structure if we reverse the signs of $d$ and all the $d_i$'s (the winding numbers of the interior defects). Therefore, the corresponding limiting solutions have the same set of interior defects. We illustrate this with the cases $d_1=1$ and $d_1=-1$. The resulting minima of $W$ have one central point defect (note that $r_N=0$ for $N=1$, $d_1=1$ ($d_1 = -1$) in Theorem~\ref{th:disk}) and the corresponding limiting solutions, $\P_*$ have one $+1/2$ ($-1/2$) central defect respectively. The limiting solutions converge to the corresponding rLdG minimizers are shown in the Fig. \ref{fig:disk_+1_-1}.
\begin{figure}[ht]
    \centering
\includegraphics[width=0.4\textwidth]{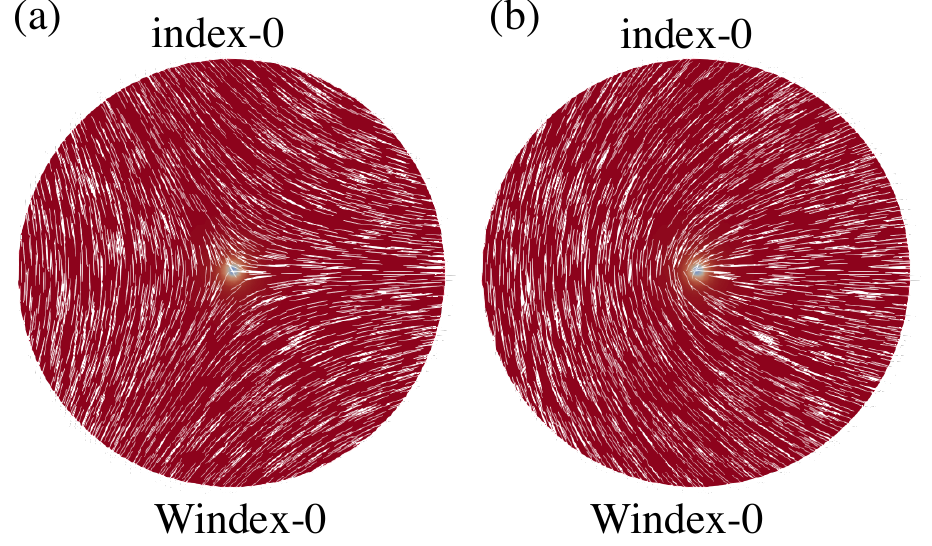}
\caption{The rLdG solutions with $\lambda^2 = 2000$ for (a) $g\times g_{\tau}\equiv -1$ and (b) $g\times g_{\tau}\equiv +1$.}
\label{fig:disk_+1_-1}
\end{figure}

Consider critical points of the rLdG energy in \eqref{eq:p_energy} with $\mathbf{d} = (1,1,1)$ and $g\times g_{\tau}\equiv3$ as an example shown in Fig. \ref{fig:disk_rotation}. These critical points are constructed from a critical point of $W$ with three interior defects, $\left\{a_1, a_2, a_3 \right\}$, each with $d_i=+1$, located at the vertices of an equilateral triangle as in Theorem~\ref{th:disk}.
According to Remark \ref{rk:rotation_invariance}, we have rotation invariance for the renormalized energy, $W$, but the corresponding limiting solutions do not have rotation invariance. Hence, there are infinitely many limiting solutions with the same relative defect positions and winding numbers. 
\begin{figure}[ht]
    \centering
\includegraphics[width=0.8\textwidth]{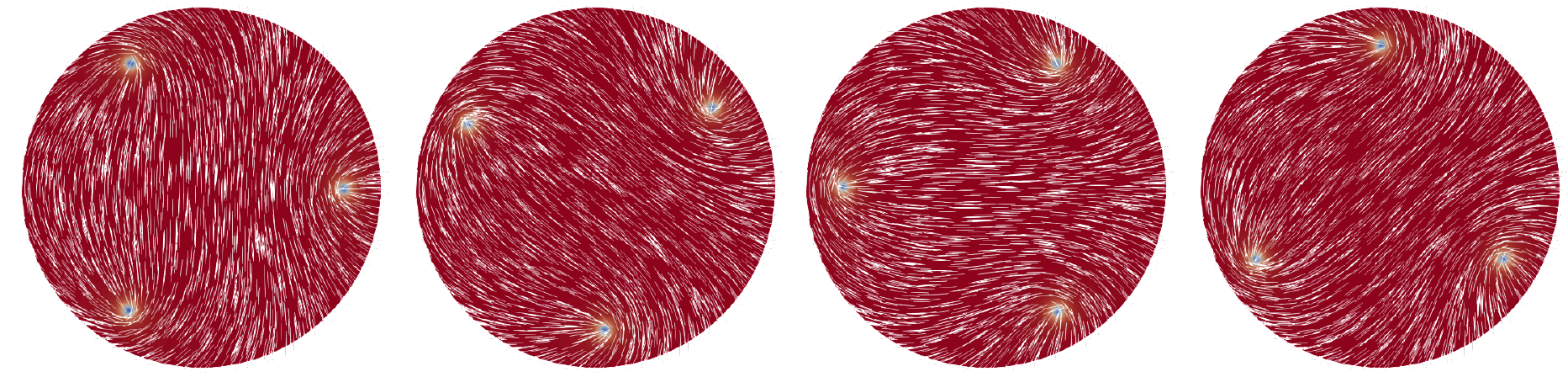}
\caption{Numerically computed rLdG solutions for $\lambda^2=2000$ and with different initial conditions, constructed from different critical points of $W$ obtained by rotating the positions of the three 
$+1/2$ defects (or defects of degree $d_i=+1$ for $u_*$ in \eqref{eq:u*}).}
\label{fig:disk_rotation}
\end{figure}

Next, consider $g\times g_{\tau}\equiv 3$ and $4$. For $g\times g_{\tau}\equiv 3$ in Fig. \ref{fig:disk_3_4}(a), we find three rLdG solutions constructed by means of three distinct limiting solutions, $\P_*$, with different critical points of $W$. For $\mathbf{d}=\left(1,1,1 \right)$, $W$ attains its minimum when the three interior defects are located at the vertices of a triangle as defined in Theorem~\ref{th:disk}. The corresponding rLdG solution is index-$0$ with three $+1/2$ defects located at the corners of a regular triangle. We also find an index-$1$ rLdG solution with one central $+1/2$ defect and two $+1/2$ defects at the symmetric positions on a diameter of the disk, and an index-$2$ rLdG solution with a $+1/2$ defect and a $+1$ defect on the diameter of the disk. By rotating the positions of these two defects, we obtain different solutions. For the $+1$ defect, the local structure near the defect may be star-like, spiral-like, or circular. To avoid repetition, a figure is not included to this effect. For $g\times g_{\tau}\equiv 4$ in Fig. \ref{fig:disk_3_4}(b), we find three rLdG solutions, again constructed by means of different limiting solutions $\P_*$ obtained via different critical points of $W$ : a minimizing solution with four $+1/2$ defects located at the corners of a square which corresponds to the $4$-polygon configuration in Theorem \ref{th:disk}; an index-$2$ rLdG solution with one $+1$ defect and two $+1/2$ defects at the corners of an isosceles triangle; and an index-$4$ rLdG solution with two $+1$ defects at symmetric positions on a diameter of the disk. 

\begin{figure}[ht]
    \centering
\includegraphics[width=\textwidth]{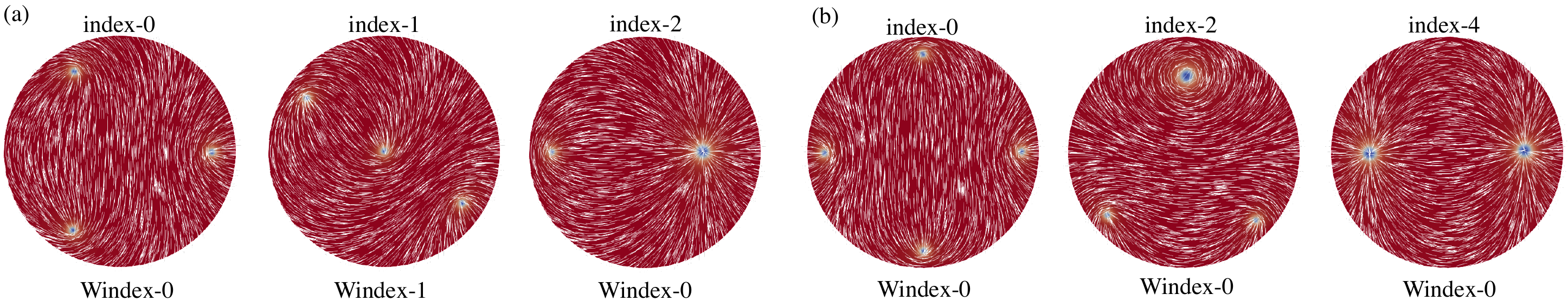}
\caption{The rLdG solutions with $\lambda^2 = 2000$ for (a) $g\times g_{\tau}\equiv 3$ and (b) $g\times g_{\tau}\equiv4$.}
\label{fig:disk_3_4}
\end{figure}

\section{Non-smooth Domain: Polygonal domains}\label{sec:polygon}

Let $\Om\subset\R^2\simeq\C$ be a bounded simply connected polygonal or curvilinear polygonal domain. Its boundary consists of finitely many open $C^2$ arcs
\[
  \Gamma_1,\ldots,\Gamma_K,
\]
meeting at vertices
\[
  V=\{b_1,\ldots,b_K\}.
\]
The interior angle at $b_j$ is denoted by $\alpha_j\in(0,2\pi)$. The Dirichlet boundary condition is
\begin{equation}\label{eq:gdatum}
  g=\rho_b e^{i\varphi_b}\quad \text{on }\partial\Om\setminus V,
\end{equation}
where $0\le \rho_b\le 1$, $\rho_b^{-1}(0)=V$, and $\varphi_b$ is a real $C^2$ lift on each open side $\Gamma_j$, $j=1,\ldots,K$. 
The modulus has algebraic zero of order $\alpha_j>0$:
\begin{equation}\label{eq:modulus-zero}
  \rho_b(x)=|x-b_j|^{\alpha_j}h_j^\pm(x),
  \qquad h_j^\pm(b_j)>0,
\end{equation}
along the two sides adjacent to $b_j$, where $h_j^{\pm}$ are $C^2$ up to $b_j$ on the two adjacent sides.
The admissible set is
\begin{equation}
H_{g}^1 = \{u\in H^1(\Omega;\mathbb{C});u = g \ on\ \partial \Omega\}.
\end{equation} 

\begin{definition}\label{def:adm}
A family $u_\eps$ of critical points of \eqref{eq:GL_energy} is deemed to be reduced admissible if, 
\begin{equation}\label{eq:uv}
u_{\eps} = v_{\eps}\rho_{\eps},
\end{equation}
$\rho_\eps$ is the scalar solution of
\begin{align}\label{eq:rho-pde}
  -\Delta\rho_\eps&=\eps^{-2}\rho_\eps(1-\rho_\eps^2),\ in\ \Omega\\
\rho_\eps&\geq 0,\ in\ \Omega\\
  \rho_\eps &= \rho_b\quad \text{on }\partial\Om\setminus V,
\end{align}
and the normalized map $v_{\eps}$ satisfies
\begin{equation}\label{eq:reduced}
  -\divop(\rho_\eps^2\nabla v_\eps)=
  \eps^{-2}\rho_\eps^4v_\eps(1-|v_\eps|^2).
\end{equation}
Hence, $v_{\eps}$ is a critical point of $F_\eps$, where
\begin{equation}\label{eq:Feps}
  F_\eps(v)=\frac12\int_\Om \rho_\eps^2|\nabla v|^2\,dx+
  \frac1{4\eps^2}\int_\Om \rho_\eps^4(1-|v|^2)^2\,dx.
\end{equation}
We assume there is $C$ independent of $\eps$ such that
\begin{equation}\label{eq:log}
  F_\eps(v_\eps)\le C\ln\frac1\eps.
\end{equation}
We also assume the
vertex logarithmic tightness condition which stipulates that  for every $j$, there exist $r_0>0$ and $C_j>0$
such that
\begin{equation}
  \int_{S_j(r_0)\setminus S_j(r)}
  \rho_\eps^2|\nabla v_\eps|^2\,d x
  \leq C_j\left(1+\ln\frac{r_0}{r}\right),
  \qquad \eps\leq r<r_0,
  \label{eq:vertex-log-tightness}
\end{equation}
where $S_j(r)$ denotes the
radius-$r$ vertex sector around $b_j$. 
This condition rules out zero-trace negative sector modes concentrating at a vertex and hence helps to identify the limit as a proper canonical map.
\end{definition}

\begin{theorem}\label{thm:main}
Let $\Omega$ and $g=\rho_be^{i\varphi_b}$ satisfy the basic assumptions above. 
Let $u_\eps$ be a family of reduced admissible  critical points of 
\begin{equation}\label{eq:energy}
  E_\eps(u)=\frac12\int_\Om |\nabla u|^2\,dx+
  \frac1{4\eps^2}\int_\Om (1-|u|^2)^2\,dx,
\end{equation} 
which is defined in Definition \ref{def:adm}.

Then there is a subsequence, $\left\{u_{\eps_\ell} \right\}$ with $\eps_\ell \to 0$ as $l\to \infty$, a set of distinct interior points $a_1,\ldots,a_N\in\Omega$, non-zero integers $d_i$, $i=1,\ldots, N$, and 
$c_j \in \mathbb{R}, j=1\ldots K$ such that
$\left\{v_{\eps_\ell}\right\}$ converges locally smoothly away from the singular set
$A=\{a_1,\ldots,a_N\}$ and the vertex set $V=\{b_1,\ldots,b_K\}$ to a canonical harmonic map
\[
  v_{\eps_\ell}\to u_*
  \quad\text{in }C^k_{\loc}(\Omega\setminus(A\cup V))
  \quad\forall k\ge0,
\] 
where $v_{\eps_\ell}$ satisfies \eqref{eq:reduced} and is related to $u_{\eps_\ell}$ as in \eqref{eq:uv} with $\eps =\eps_\ell$.

The canonical harmonic map is labelled by the set $\mathcal C=(\mathbf{a},\mathbf{c},\mathbf{d})$.
We quantify the corner contribution of winding number $c_j$ by the negative value of the rotation of the director field around the corner $b_j$ between the two adjacent edges $\Gamma_{j+1}$ and $\Gamma_j$. From the Poincaré-Hopf Theorem, $\mathbf{d}$ and $\mathbf{c}$ satisfy
\begin{equation}\label{eq:topological constraint}
  \sum_{i=1}^N d_i
  =  \frac1{2\pi}\left(
  \sum_{j=1}^K\int_{\Gamma_j}\partial_\tau\varphi_b\,ds\right) + \sum_{j =1}^K c_j.
\end{equation}
The limiting harmonic map, $u_*$, is then given by \begin{equation}\label{eq:main-canonical}
  u_*(z)=\e^{i\phi}
  \prod_{i=1}^N\left(\frac{z-a_i}{|z-a_i|}\right)^{d_i}
  \prod_{j=1}^K\left(\frac{z-b_j}{|z-b_j|}\right)^{-\frac{2\pi c_j}{\alpha_j}},
\end{equation}
where $\alpha_j$ is the inner angle of $b_j$, 
and $\phi$ is the harmonic solution of the zero-degree boundary problem in 
\begin{equation}\label{eq:H}
\begin{cases}
  \Delta \phi = 0,
  & \text{in }\Omega,\\[2mm]
  \phi = \phi_0,
  & \text{on }\partial\Omega.
\end{cases}
\end{equation} 
where

\begin{equation}
e^{i\phi_0(z)} = e^{i\varphi_b(z)}\frac{|z-a_1|^{d_1}}{(z-a_1)^{d_1}}\ldots\frac{|z-a_N|^{d_N}}{(z-a_N)^{d_N}}\frac{|z-b_1|^{-\frac{2\pi c_1}{\alpha_1}}}{(z-b_1)^{-\frac{2\pi c_1}{\alpha_1}}}\ldots\frac{|z-b_K|^{-\frac{2\pi c_K}{\alpha_K}}}{(z-b_K)^{-\frac{2\pi c_K}{\alpha_K}}}
\end{equation} and
 $u_*(z) = e^{i\phi_0}\frac{(z-a_1)^{d_1}}{|z-a_1|^{d_1}}\ldots\frac{(z-a_N)^{d_N}}{|z-a_N|^{d_N}} \frac{(z-b_1)^{-\frac{2\pi c_1}{\alpha_1}}}{|z-b_1|^{-\frac{2\pi c_1}{\alpha_1}}}\ldots\frac{(z-b_K)^{-\frac{2\pi c_K}{\alpha_K}}}{|z-b_K|^{-\frac{2\pi c_K}{\alpha_K}}}=e^{ i \varphi_b(z)}$ on $\partial \Omega$. 

In addition, we have 
\begin{equation}
\nabla\left(u_*(z)\frac{|z-a_j|^{d_j}}{(z-a_j)^{d_j}}\right)(a_j) = 0,\  j = 1,\ldots,N,
\end{equation}
which expresses that the vector of interior defect locations, $(a_1,\ldots,a_N)$, is a critical point of some appropriate renormalized energy, i.e.,
\begin{equation}\label{eq:criticality}
  \nabla_{a_i}W(\mathcal C)=0,
  \qquad i=1,\ldots,N.
\end{equation}
Here the renormalized energy is defined by 
\begin{align}
\label{eq:W_b}
  W(\mathcal C)
  &=\frac12 \int_{\partial\Om\backslash V}\Phi_0q\,ds
    -\pi\sum_{i=1}^Nd_iR_0(a_i)
    +\pi\sum_{j=1}^K c_jR_0(b_j) \\
&-\pi \sum_{1\leq i\neq j\leq N} d_i d_j \ln|a_i-a_j|  - \pi \sum_{1\leq i\neq j\leq K} \frac{2\pi}{\alpha_j} c_i c_j \ln|b_i-b_j|\\
&+ \pi\sum_{i = 1}^N\sum_{j = 1}^K \left(\frac{2\pi}{\alpha_j}+1\right) d_i c_j \ln|b_j-a_i|\label{eq:W_b_end}
\end{align}
where \begin{equation}
\label{eq:Phi0-definition}
  \Phi_0(x)=\sum_{i=1}^N d_i\ln|x-a_i|-
  \sum_{j=1}^K\frac{2\pi c_j}{\alpha_j}\ln|x-b_j|+R_0(x),
\end{equation}
and
\begin{align}
&\Delta R_0 = 0,\ \text{on}\ \Omega\label{eq:108}\\
&\frac{\partial R_0}{\partial \nu} =  q + \sum_{j = 1}^K \frac{2\pi}{\alpha_j} c_j \frac{\partial \ln|z-b_j|}{\partial \nu} - \frac{\partial \sum_{i = 1}^N d_i \ln|z-a_i|}{\partial \nu}\ \text{on}\ \partial \Omega\backslash V,\label{eq:108-1}
\end{align}
and
\begin{equation}
q:=\e^{i\varphi_b}\times\partial_\tau\e^{i\varphi_b}=\partial_\tau\varphi_b\ on\ \partial\Om\setminus V.
\end{equation}
\end{theorem}

\begin{remark}The proof of Theorem \ref{thm:main} can be found in the preprint \cite{han2026asymptotictheoryginzburglandaucritical}.
\end{remark}

\begin{remark} Theorem \ref{thm:main} shows that reduced admissible critical points of the GL energy can be characterized by critical points of a suitably defined renormalized energy, in the $\eps \to 0$ limit. 
The critical points of $W$ might not correspond to a critical point of GL energy. In what follows, we follow the same strategy as in Section~\ref{sec:disk}, i.e., use the critical points of $W$ on a polygon to construct limiting harmonic maps $u_*$ of the form \eqref{eq:main-canonical}, which then yield high-quality initial conditions, $\P_*$ for the rLdG solver via the relations \eqref{eq:uP}. 
\end{remark}
\begin{remark} The renormalized energy for polygonal or curvilinear polygonal domains in \eqref{eq:W_b} can be reduced to the renormalized energy for smooth domains in \eqref{eq:Wa} by eliminating the $b_j$ terms and by removing all consideration to the vertex set $V$ in the boundary integrals. 
\end{remark}

\begin{proposition}\label{prop:tangentialrenormalizedenergy}
For tangential or homeotropic boundary conditions on a polygon with straight edges (so that the argument of $\varphi_b$ or equivalently $\gamma$ in \eqref{eq:uP} is constant on a straight line), the renormalized energy $W$ in \eqref{eq:W_b}-\eqref{eq:W_b_end} can be written as 
\begin{align} 
W(\mathcal C) =& -\pi \sum_{1\leq i\neq j\leq N} d_i d_j \ln \left| (1 - w_i \overline{w}_j)(w_i - w_j) \right| \label{eq:l1}\\
& -\pi \sum_{i=1}^N d_i^2 \ln \left( 1 - |w_i|^2 \right) \label{eq:l2}\\
&+ \pi \sum_{i=1}^N \sum_{m=1}^K d_i^2 \left(\frac{\alpha_m}{\pi} - 1\right) \ln \left| w_i - w^{(c)}_m \right|  + \pi \ln|C_{SC}
| \sum_{i=1}^N d_i^2  \label{eq:l3}\\
& + 4\pi \sum_{i=1}^N \sum_{m=1}^K d_i c_m \ln \left| w_i - w^{(c)}_m \right|  \label{eq:l4}\\
& - 2\pi \sum_{1\leq m\neq l\leq K} c_m c_l \ln \left| w^{(c)}_m - w^{(c)}_l \right|  \label{eq:l5}\\
&+ \sum_{m=1}^K \frac{2\pi^2}{\alpha_m} c_m^2 \ln |\mathcal{A}_m|.\label{eq:l7}
\end{align} 
This renormalized energy only depends on the pre-image locations, $\mathbf{w} = f^{-1}(\mathbf{a})$, of the interior defects and the fixed pre-images, $\mathbf{w}^{(c)} = f^{-1}(\mathbf{b})$, of the polygon vertices, where $f$ is the Schwarz-Christoffel mapping, $z = f(w)$, that transforms the calculation on the polygonal domain $\Omega$ in the $z$-plane to the unit disk $D$ in the $w$-plane. 
\end{proposition}
\begin{proof}
For tangential or homeotropic boundary conditions, we have $q =\partial_{\tau}\varphi_b\equiv 0$ on $\partial\Omega\backslash V$, where $\Omega$ is a two-dimensional polygon domain. To compute $R_0(a_i)$ and $R_0(b_j)$, we introduce the Schwarz-Christoffel mapping
\begin{equation}\label{eq:f}
f(w) = C_{SC} \int_{0}^{w} \prod_{m=1}^K \left(1 - \frac{\zeta}{w_m^{(c)}}\right)^{\frac{\alpha_m}{\pi} - 1} d\zeta,
\end{equation}
and 
\begin{equation}\label{eq:f'}
f'(w) = C_{SC} \prod_{m=1}^K \left(1 - \frac{w}{w_m^{(c)}}\right)^{\frac{\alpha_m}{\pi} - 1},
\end{equation} 
where $w \in D$, $f(w) \in \Omega$, the complex constant $C_{SC}$ determines the scale and rotation of the polygon,  $\alpha_m$ are the interior angles of the polygon vertices and $w_m^{(c)}$ are the pre-vertices on the unit circle $\partial D$.

We define $\tilde{R}_0$ such that
\begin{equation}\label{eq:R0tildeR0}
\tilde{R}_0(w) = \tilde{R}_0\bigl(f^{-1}(z)\bigr)=R_0(z), \quad z\in \Omega.
\end{equation}
Since $f$ is a conformal map from the unit disk to the polygon, $\tilde{R}_0$ is a solution of the Laplace equation on the disk. Substituting the relation between $R_0$ and $\tilde{R}_0$ in \eqref{eq:R0tildeR0}, into \eqref{eq:108}-\eqref{eq:108-1}, we have
\begin{align}
&\Delta \tilde{R}_0 = 0,\ \text{on}\ D\\
&\frac{\partial \tilde{R}_0}{\partial \mu} =  \sum_{j = 1}^K \frac{2\pi}{\alpha_j} c_j \frac{\partial \ln|f(w)-b_j|}{\partial \mu} - \frac{\partial \sum_{i = 1}^N d_i \ln|f(w)-a_i|}{\partial \mu}\ \text{on}\ \partial D\backslash\{w_1^{(c)},\ldots,w_K^{(c)}\}.
\end{align}
where $\mu$ is the outward normal direction on the boundary of the disk, $a_i = f(w_i)$, $i = 1,\ldots,N$ $b_j = f(w_j^{(c)})$, $j = 1,\ldots, K$. According to Proposition \ref{prop:tildeR0} in Appendix B, the solution of the boundary-value problem for $\tilde{R}_0$ is given by 
\begin{equation}\label{eq:tildeR0new}
\tilde{R}_0(w) = \sum_{j=1}^K \frac{2\pi}{\alpha_j} c_j \ln \left| \frac{f(w) - b_j}{(w - w_j^{(c)})^{\alpha_j/\pi}} \right| + \sum_{i=1}^N d_i \ln \left| \frac{(1 - w \bar{w}_i)(w - w_i)}{f(w) - a_i} \right|, \quad w\in D.
\end{equation}

To evaluate $R_0(a_k) = \tilde{R}_0(w_k)$, we take the limit as $w \to w_k$ and subsequently the singular term in the second sum in \eqref{eq:tildeR0new} reduces to
\[
\lim_{w\to w_k}d_k \ln \left| \frac{(1 - w \bar{w}_k)(w - w_k)}{f(w) - a_k} \right| = d_k \ln \left| \frac{1 - |w_k|^2}{f'(w_k)} \right|.
\]
Substituting \eqref{eq:f'} into the above equation, we have
\begin{align}
&R_0(a_k)=\tilde{R}_0(w_k) = d_k \ln \left| \frac{1 - |w_k|^2}{f'(w_k)} \right| + \sum_{\substack{i=1 \\ i \neq k}}^N d_i \ln \left| \frac{(1 - w_k \bar{w}_i)(w_k - w_i)}{a_k - a_i} \right| +  \sum_{j=1}^K \frac{2\pi}{\alpha_j}c_j \ln \left| \frac{f(w_k) - b_j}{(w_k - w_j^{(c)})^{\alpha_j/\pi}} \right|,\\
&= d_k \ln \left| \frac{1 - |w_k|^2}{C_{SC} \prod_{m = 1}^K(1 - w_k/w_m^{(c)})^{\frac{\alpha_m}{\pi}-1}} \right| + \sum_{\substack{i=1 \\ i \neq k}}^N d_i \ln \left| \frac{(1 - w_k \bar{w}_i)(w_k - w_i)}{a_k - a_i} \right| +  \sum_{j=1}^K \frac{2\pi}{\alpha_j}c_j \ln \left| \frac{a_k - b_j}{(w_k - w_j^{(c)})^{\alpha_j/\pi}} \right|,\\
&= d_k \ln \left| \frac{1 - |w_k|^2}{C_{SC} \prod_{m = 1}^K(w_m^{(c)} - w_k)^{\frac{\alpha_m}{\pi}-1}} \right| + \sum_{\substack{i=1 \\ i \neq k}}^N d_i \ln \left| \frac{(1 - w_k \bar{w}_i)(w_k - w_i)}{a_k - a_i} \right| +  \sum_{j=1}^K \frac{2\pi}{\alpha_j}c_j \ln \left| \frac{a_k - b_j}{(w_k - w_j^{(c)})^{\alpha_j/\pi}} \right|.
\end{align}
Take the $w \to w_m^{(c)}$ limit to evaluate $R_0(b_m) = \tilde{R}_0(w_m^{(c)})$. Since $|w_m^{(c)}| = 1$ we have $|1 - w_m^{(c)} \bar{w}_i| = |w_m^{(c)} - w_i|$. To deal with the first term in 
\eqref{eq:tildeR0new}, we define the non-zero asymptotic mapping constant $\mathcal{A}_m$,
\begin{align*}
\mathcal{A}_m &= \lim_{w \to w_m^{(c)}} \frac{f(w) - b_m}{(w - w_m^{(c)})^{\alpha_m/\pi}}= \lim_{w \to w_m^{(c)}} \frac{\frac{d}{dw} \left( C_{SC} \int_{w_m^{(c)}}^{w} (1 - \frac{\zeta}{w_m^{(c)}})^{\frac{\alpha_m}{\pi}-1} d\zeta \right)}{\frac{d}{dw}(w - w_m^{(c)})^{\alpha_m/\pi}}\\
&= \frac{ C_{SC}\pi}{\alpha_m} \lim_{w \to w_m^{(c)}} \left( \frac{1 - \frac{w}{w_m^{(c)}}}{w-w_m^{(c)}} \right)^{\frac{\alpha_m}{\pi}-1}= \frac{ C_{SC}\pi}{\alpha_m} ( - w_m^{(c)})^{1-\frac{\alpha_m}{\pi}}.\end{align*}

The regularized field at the corner, $f(w_m^{(c)}) = b_m$, is given by:
\begin{align}
R_0(b_m)=\tilde{R}_0(w_m^{(c)}) = \frac{2\pi}{\alpha_m} c_m \ln |\mathcal{A}_m| +  \sum_{\substack{j=1 \\ j \neq m}}^K \frac{2\pi}{\alpha_j}c_j \ln \left| \frac{f(w_m^{(c)}) - b_j}{(w_m^{(c)} - w_j)^{\alpha_j/\pi}} \right| + \sum_{i=1}^N d_i \ln \left| \frac{(w_m^{(c)} - w_i)^2}{f(w_m^{(c)}) - a_i} \right|.
\end{align}

Substituting $R_0(a_k)$ and $R_0(b_m)$ into the renormalized energy $W$ in \eqref{eq:W_b}, the logarithmic map terms $\ln|a_k - a_i|$, $\ln|a_k - b_j|$, $\ln|b_m - b_j|$ perfectly cancel, so that $W$ only depends on the pre-image locations $\mathbf{w}$ of the interior defects and the fixed pre-images $\mathbf{w}^{(c)}$ of the polygon corners.
\end{proof}
The role of each term in the renormalized energy in Theorem~\ref{thm:main} can be explained as follows.
\eqref{eq:l1} is the standard defect-defect interaction energy. It includes both the direct distance $|w_i - w_j|$ and the interaction of defects with the mirror images of the remaining defects, $|1 - w_i \bar{w}_j|$. \eqref{eq:l2} is the standard self-energy of interior defects in a disk. \eqref{eq:l3} is the geometric correction arising from the Schwarz-Christoffel mapping derivative $f'(w)$, modelling the strong repulsion/attraction induced individually by each polygon corner $m$ proportional to its inner angle $\alpha_m$. \eqref{eq:l4} models the interaction between an interior defect and a polygon corner or vertex. Remarkably, the $4\pi$ coefficient is a topological invariant independent of the polygon's shape. \eqref{eq:l5} models the interaction energy between vertices or corner defects. The self-energy of corner defects in \eqref{eq:l7} is evaluated explicitly for each corner's distinct local geometry.

\section{Numerical Results}\label{sec:numericspolygons}

There is substantial literature on critical points of the rLdG energy \eqref{eq:p_energy} on polygons, subject to Dirichlet tangential boundary conditions \cite{han2021solution, shi2022nematic, yin2020construction}, primarily using qualitative analyses of the Euler-Lagrange equations and finite-element or finite difference methods. In what follows, we compute critical points (optimized interior defect locations) of the renormalized energy in Theorem~\ref{thm:main} for different combinations of interior and corner defect degrees and use the critical points of the renormalized energy to locate corresponding critical points of the rLdG energy (when they exist) for canonical examples. The computation follows a framework analogous to that of the disk domain, but with essential modifications to accommodate defects at the corners, the different formulation of the renormalized energy and the combinatorial selection of both interior defect winding number contribution, $\mathbf{d}$, and corner defect winding number contribution, $\mathbf{c}$. We take the hexagon as an example of a regular domain and take the rectangle  as an example of an irregular geometry, and compare the numerical results with our previous studies. In fact, by using the Schwarz-Christoffel mapping $f$, we are able to deal with any polygons, regular or irregular, convex or concave.

\subsection{Hexagon Domain}\label{sec:hexagon}
Consider a regular hexagon with tangential boundary conditions for the nematic director. Recall that the rLdG order parameter in two dimensions is given by
\[
\P = s \left(2 \mathbf{n}\otimes \mathbf{n} - \mathbf{I}_2 \right)
\] and the tangential boundary conditions require $\mathbf{n}$ (the nematic director) to be tangential to the polygon edges. Equivalently, if $\mathbf{n}=\left(\cos\gamma, \sin\gamma \right)$, then $\gamma$ is constant on each edge of the hexagon. 

We now provide a systematic paradigm for computing critical points of the renormalized energy in Proposition \ref{prop:tangentialrenormalizedenergy}. For a regular hexagon, we have that $K=6$, the vertices are located at  $b_j = e^{\frac{i2\pi j}{K}}$ with inner angle $\alpha_j = \frac{(K-2)\pi}{K}$, $j = 1,\ldots,K$. There are two types of vertices dictated by the tangential boundary conditions - \emph{splay} vertex for which $\mathbf{n}$ is radial around the vertex or $\mathbf{n}$ splays near a splay vertex and \emph{bend} vertex for which $\mathbf{n}$ bends around a vertex. The winding number contribution of a bend vertex is $c_j = 1/3$ whereas the winding number contribution of a splay vertex is $c_j = -2/3$. The corresponding degrees of the nematic director are $-c_j^P \in \left\{-1/6, 1/3 \right\}$, since the topological degree of the nematic director is half of the degree of $u_*$ in \eqref{eq:main-canonical}. Thus, the critical points of $W$ have the simplest structure when there are two splay vertices, four bend vertices and there are no interior defects since
\[
\sum_{i=1}^N d_i = \sum_{j=1}^{K} c_j.
\] These critical points of $W$ correspond to the interior defect-free critical points of the GL energy or the rLdG energy \eqref{eq:p_energy}, labelled as the \emph{Para}, \emph{Meta} and \emph{Ortho} critical points or rLdG solutions in \cite{han2021solution}, each of which has two splay, i.e., $n_s = 2$ and four bend vertices, i.e., $n_b = 4$. There are three rotationally equivalent \emph{Para} solutions, each with a pair of diametrically opposite splay vertices. The \emph{Para} solutions are the rLdG energy minimizers for sufficiently small $\eps$ or large $\lambda$, according to numerical computations in \cite{han2021solution}. The \emph{Meta} solutions have two splay vertices separated by a bend vertex and the \emph{Ortho} solutions have two adjacent splay vertices. Unsurprisingly, whilst the \emph{Para}, \emph{Meta}, \emph{Ortho} solutions can be index-$0$ critical points of the rLdG energy, the \emph{Ortho} solutions have the highest energy in this class of solutions with no interior defects. In what follows, we use the renormalized energy in Theorem~\ref{thm:main} to recover existing results in the literature, i.e., recover the reported critical points of the rLdG energy on a hexagon with tangential boundary conditions in \cite{han2021solution} and to discover new critical points of the rLdG energy, that might be missed by conventional numerical methods.

To systematically explore all physically permissible configurations on a regular hexagon with tangential boundary conditions, we must determine both the vertex/corner configuration array $\mathbf{c} = (c_1, c_2, \dots, c_K)$ and the interior defect array $\mathbf{d} = (d_1, d_2, \dots, d_N)$. The selection strategy proceeds as follows:
\begin{enumerate}
    \item \textbf{Corner Degrees ($c_j$):}
    We constrain the vertices to be either bend or splay vertices, so that $c_j \in \left\{1/3, -2/3 \right\}$. We allow for any geometrically allowed number of splay and bend vertices, $n_s$ and $n_b$ such that $n_s + n_b = 6$. 
    For a hexagon ($K=6$), this yields $2^6 = 64$ possible boundary combinations.
    \item \textbf{Cyclic Symmetry Reduction:} To avoid redundant computations of rotationally equivalent states, we compute the cyclic permutations of all $64$ combinations and extract only the $14$ unique cyclic equivalence classes.
    \item \textbf{Internal Degrees ($d_i$) and Topological Constraint:} For each corner configuration, we compute its total boundary degree $d = \sum_{j=1}^K c_j$. We then iterate through combinations of interior defect degrees $d_i \in \{1, 2, -1, -2\}$ with replacement. A combination is deemed valid if and only if it satisfies the topological constraint in \eqref{eq:topological constraint},
    $\sum_{i=1}^N d_i = \sum_{j=1}^K c_j$. 
\end{enumerate}

Applying the aforementioned strategy to a hexagon ($K=6$) with up to $N_{max}=4$ interior defects, the algorithm first reduces the $2^6 = 64$ possible boundary sequences to 14 unique cyclic configurations. For each of these 14 combinations of the vertex degrees, the algorithm searches for corresponding interior defect arrays $\mathbf{d}$ of length $1\le N \le N_{max}$ that satisfy the topological sum constraint above, with $|d_i|\leq 2$. This rigorous combinatorial matching yields exactly $94$ independent $(\mathbf{d},\mathbf{c})$ topological task configurations. 
In one case, we specifically set $N=7$, 
$\mathbf{c}=\left(-2/3,-2/3,-2/3,-2/3,-2/3,-2/3\right)$ (six splay vertices), and 
$\mathbf{d}=\left(2,-1,-1,-1,-1,-1,-1\right)$ to find the optimized defect locations of the seven interior defects, and the corresponding critical point of $W$ yields the high-index H critical point of the rLdG energy. The H-critical point of the rLdG energy has a central point defect of degree $+1$, surrounded by six symmetric $-1/2$ point defects in the nematic director.

Next, we compute the critical points of the renormalized energy for each admissible combination, $\left(\mathbf{d}, \mathbf{c} \right)$ (with different Morse indices) and use the resulting optimized defect locations, $\left(\mathbf{a}\right)$, to define the canonical harmonic map $u_*$ in \eqref{eq:main-canonical}. The canonical map $u_*$ is used to define high-quality initial conditions for the rLdG solver. In Fig. \ref{fig:old}, we use this strategy to recover the Ring, BD, M1, T, T10, T130, T135, TD, and H critical points of the rLdG energy, as reported in \cite{han2021solution}. We do not discuss Para, Meta and Ortho, for which the number of interior defects is zero, i.e., the renormalized energy is independent of the locations of defects $\mathbf{a}$. 

\begin{figure}[ht]
    \centering
\includegraphics[width=0.8\textwidth]{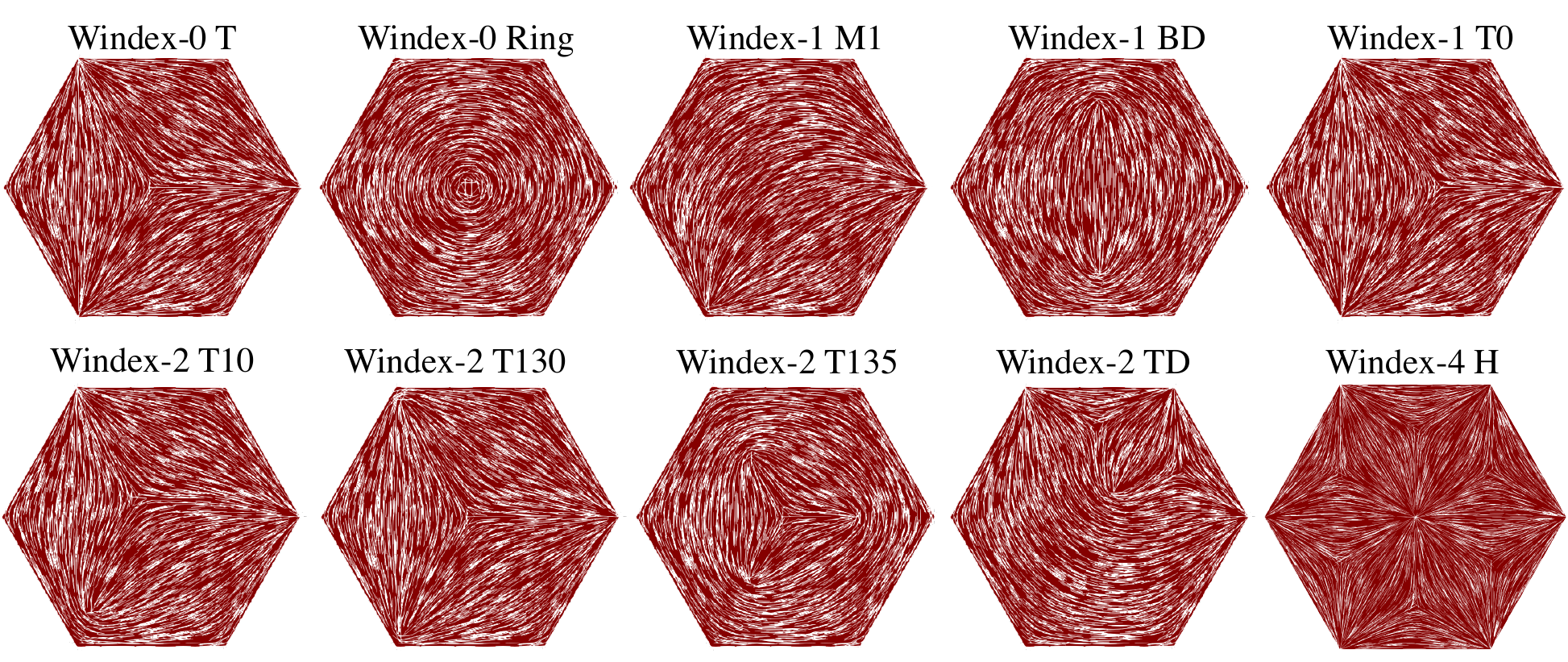}
    \caption{The profiles of the limiting harmonic maps, $\P_*$, constructed from the canonical harmonic map, $u_*$ in \eqref{eq:main-canonical}, for admissible $\left(\mathbf{d}, \mathbf{c} \right)$. These limiting solutions have corresponding rLdG solutions as reported in \cite{han2021solution}.}
    \label{fig:old}
\end{figure}

\begin{table}[htbp]
\centering
\caption{Comparison of the Morse indices of critical points for the renormalized energy and the corresponding rLdG solutions at $\lambda^2 = 600$, as reported in \cite{han2021solution}.}
\label{tab:index_comparison}
\resizebox{\textwidth}{!}{ 
\begin{tabular}{lcccccccccc}
\hline
\textbf{Configuration} & T & Ring & M1 & BD & T0 & T10 & T130 & T135 & TD & H \\ \hline
Windex    & 0 & 0    & 1  & 1  & 1  & 2   & 2    & 2    & 2  & 4 \\
index ($\lambda^2=600$) & 0 & 2    & 1  & 2  & 1  & 2   & 3    & 3    & 3  & 8 \\ \hline
\end{tabular}    
} 
\end{table}

As observed from Table \ref{tab:index_comparison}, the Morse indices of the rLdG solutions are always larger than the corresponding critical points of the renormalized energy. An intuitive explanation is that the renormalized energy can be regarded as an asymptotic limit of the rLdG model as $\lambda^2 \to \infty$, under the constraints of a fixed number of interior defects $N$, specific winding numbers $\mathbf{d}$, and prescribed boundary rotations or vertex degrees on a polygon $\mathbf{c}$. Consequently, the renormalized energy focuses on a lower-dimensional manifold of configurations, effectively reducing the available degrees of freedom for perturbations and leading to fewer unstable directions. However, it should be noted that the Morse indices of rLdG solutions depend on the parameter $\lambda^2$. In the regime of smaller $\lambda^2$, when the core structure of defects becomes more prominent and the asymptotic approximation weakens, this comparative conclusion may no longer hold.

We use the limiting solutions (constructed from the canonical harmonic map in \eqref{eq:main-canonical} and the critical points of the renormalized energy) as initial conditions to solve the Euler-Lagrange equation in \eqref{eq:rLdG_EL} by finite element method with Newton solver, and find six novel critical points of the rLdG energy in Fig. \ref{fig:new}. In Fig. \ref{fig:new}(a), the two $+1/2$ defects are displaced relative to the locations of two defects in the BD solution in Fig.~\ref{fig:old}, in the direction normal to the line connecting them. In Fig.~\ref{fig:new}(c), the three $+1/2$ defects are located closer to the corresponding corners than in the T135 solution, as plotted in Fig.~\ref{fig:old}.
\begin{figure}[ht]
    \centering
\includegraphics[width=\textwidth]{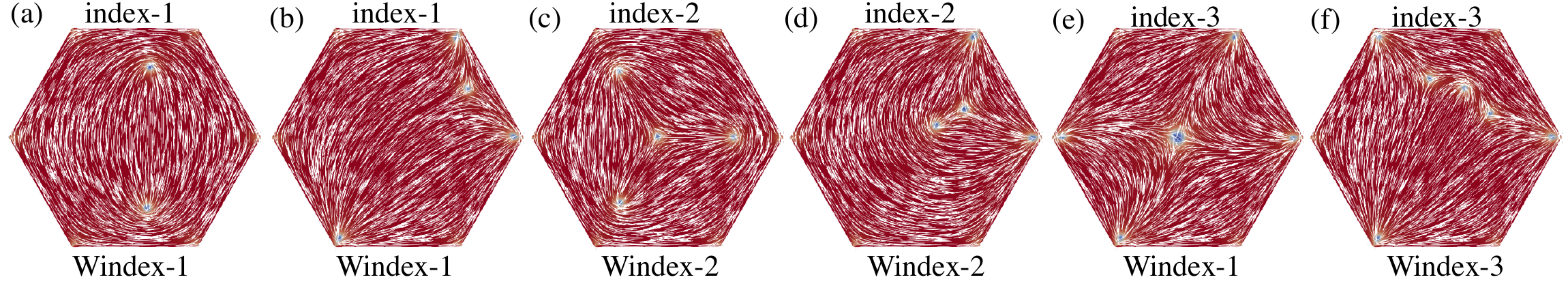}
    \caption{The profiles of novel rLdG solutions with $\lambda^2=2000$. The point defects are small blue dots with $\P \approx 0$. In (b), the rLdG solution has three splay vertices and an interior $-1/2$ defect. In (d), the rLdG solution has two adjacent splay vertices, four bend vertices and an interior $\pm 1/2$ defect pair. In (e), there are 4 splay vertices and an interior $-1$-degree defect. In (f), there are two interior $-1/2$ defects and an interior $+1/2$ defect, accompanied by three splay and three bend vertices.}
    \label{fig:new}
\end{figure}

\subsection{Rectangular Domain}\label{sec:rectangle}  
Consider a rectangle, $R=\left\{(x,y)\in \mathbb{R}^2| -L_x \leq x \leq L_x, -L_y \leq y \leq L_y \right\}$, with four vertices at $b_1 = L_x - iL_y$, $b_2 = L_x + iL_y$, $b_3 = -L_x + iL_y$, $b_4 = -L_x - iL_y$, the aspect ratio $b_{ratio} = \frac{L_y}{L_x}$ and the interior angles are $\alpha_j = \pi/2$ for $j = 1, 2, 3, 4$. As before, we work with tangential boundary conditions and the tangential boundary conditions require that the vertices are either \emph{splay} or \emph{bend} vertices. The effective topological degree at a splay vertex is $c_j = -1/2$, and $c_j = 1/2$ for a bend vertex. We restrict ourselves to $c_j \in \left\{ -1/2, 1/2 \right\}$, in what follows.

As before, we have the topological constraint
\[
\sum_{i=1}^{N}d_i = \sum_{j=1}^{4} c_j.
\] The simplest topology with no interior defects, i.e., $d_i = 0$, $i = 1,\ldots,N$ is attained when we have two splay vertices and two bend vertices. For sufficiently large $\lambda$ or sufficient small $\eps$ in \eqref{eq:GL_energy}, there are two known index-$0$ critical points of the rLdG energy - the \emph{diagonal} solution with two diagonally opposite splay vertices and the \emph{rotated} solution with two adjacent splay vertices. In what follows, we use the critical points of the renormalized energy for admissible $\left(\mathbf{d}, \mathbf{c} \right)$, subject to the topological constraint, and then use these critical points of $W$ to compute corresponding critical points of the rLdG energy on a rectangle with tangential boundary conditions.

Without any constraint on the number of splay or bend vertices, to optimize the computational search and eliminate physically redundant configurations, we account for the intrinsic $D_2$ symmetry of the rectangle. Specifically, the $\mathbf{c}$ arrays are treated as isomorphic and subsequently deduplicated if they can be mapped onto one another via a $180^\circ$ rotation or reflections across the horizontal and vertical axes. This symmetry reduction yields exactly $7$ unique, non-isomorphic configurations for the vertex charges. For the interior topological defects, denoted by $\mathbf{d}$, we restrict the individual charges to $d_i \in \{1, 2, -1, -2\}$ and limit the maximum number of defects to $N \leq 4$. By coupling these interior defect configurations with the $7$ unique vertex configurations or choices of $\mathbf{c}$—subject to the global topological constraint in \eqref{eq:topological constraint}—we generate a total of 49 distinct $(\mathbf{d}, \mathbf{c})$ combinations to serve as the initial parameter configurations for the numerical searches of critical points of the rLdG energy.

We compare our results with the critical points of the rLdG energy on a rectangular domain with tangential boundary conditions, as reported in \cite{shi2022nematic}. We recover all the rLdG solutions reported in \cite{shi2022nematic}, with point defects, for $\lambda^2 = 50$ as shown in Fig. \ref{fig:rectangle_old} and miss rLdG solutions with line defects. The renormalized energy cannot capture candidate critical points with line defects. For $\lambda^2 = 1000$, we find three new index-2 rLdG solutions in Fig. \ref{fig:rectangle_new}.
 The first solution (Fig.~\ref{fig:rectangle_new}(a)) features a $-1$ defect at its center, accompanied by two $+1/2$ defects approximately situated along the diagonal. The interior defects have a total of zero topological charge, with two splay and two bend vertices. The second solution (Fig.~\ref{fig:rectangle_new}(b)) exhibits a single $-1$ defect located at the center, accompanied by four splay vertices. The third rLdG solution (Fig.~\ref{fig:rectangle_new}(c)) is characterized by a central $+1$ defect, accompanied by four bend vertices. It is interesting that the index of the corresponding critical point of the renormalized energy is Windex-$0$ for the second and third solutions, i.e., the localized central point defect corresponds to the minimum of the renormalized energy when $c_j$ is the same for all four vertices and $N=1$. These minima of $W$ are used to define a canonical harmonic map in \eqref{eq:main-canonical}, with an associated limiting $\P_*$ which converge to the rLdG critical points in Fig.~\ref{fig:rectangle_new}. The rLdG critical points have a higher Morse index owing to greater freedom in the potential directions of instability.

\begin{figure}[ht]
    \centering
\includegraphics[width=\textwidth]{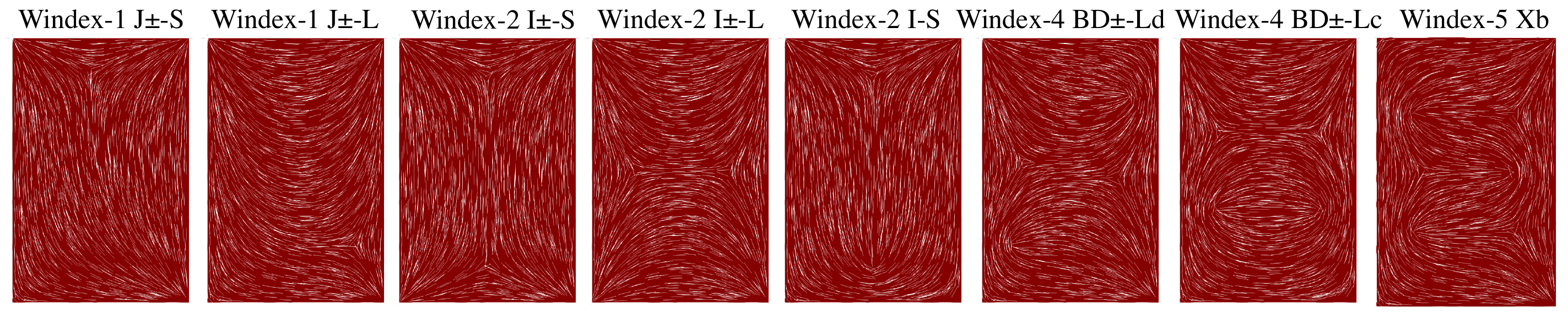}
    \caption{The profiles of limiting solutions, $\P_*$, corresponding to the rLdG solutions reported in \cite{shi2022nematic}.}
    \label{fig:rectangle_old}
\end{figure}

\begin{table}[htbp]
\centering
\caption{Comparison of the Morse indices of critical points for the renormalized energy and the corresponding rLdG solutions at $\lambda^2 = 50$, as reported in \cite{shi2022nematic}.}
\label{tab:index_comparison_rectangle}
\resizebox{\textwidth}{!}{ 
\begin{tabular}{lcccccccc}
\hline
\textbf{Configuration} & J$\pm$-S & J$\pm$-L & I$\pm$-S & I$\pm$-L & I-S & BD$\pm$-Ld & BD$\pm$-Lc & Xb  \\ \hline
Windex    & 1 & 1    & 2  & 2  & 2  & 4   & 4    & 5   \\
index ($\lambda^2=50$) & 1 & 1    & 2  & 2  & 2  & 4   & 4    & 6   \\ \hline
\end{tabular}    
} 
\end{table}

\begin{figure}[ht]
    \centering
\includegraphics[width=0.4\textwidth]{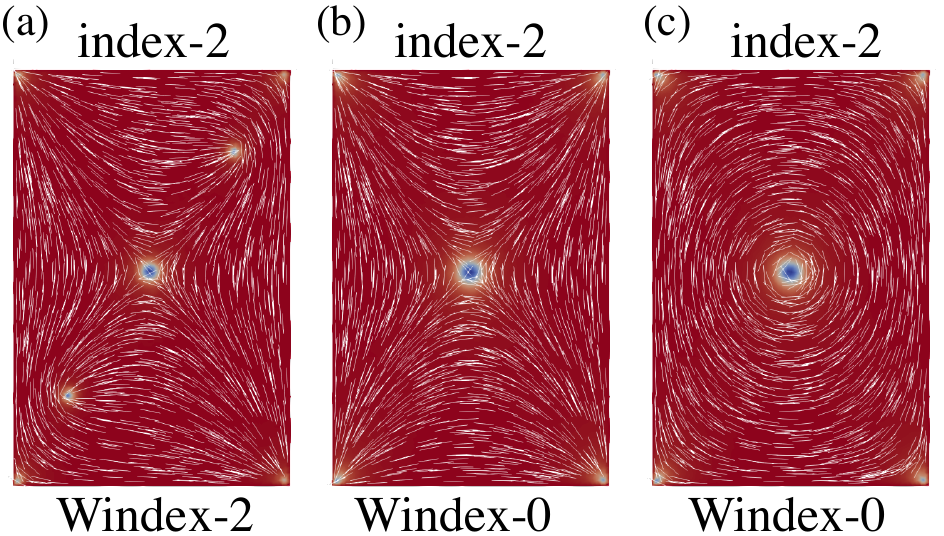}
    \caption{The profiles of novel rLdG solutions with $\lambda^2=1000$.}
    \label{fig:rectangle_new}
\end{figure}

In Table \ref{tab:index_comparison} and Table \ref{tab:index_comparison_rectangle}, we find that when comparing the Windex of a limiting solution with the index of its corresponding rLdG solution in a domain of finite size, the rLdG index is always greater than or equal to the Windex. 
Additional negative directions
may come from splitting instability of interior defects and
unstable modes inside the vertex sectors. In addition, practically, finite core-size effects may change the Morse index at small but finite $\eps$ or modest values of $\lambda$. When Windex =index, we can use the explicitly computable instabilities of the critical points of the renormalized energy to identify the unstable eigendirections of the corresponding rLdG solutions. This could provide novel control strategies or penetrating insight into how to stabilize high-index saddle points of the rLdG energy on bounded domains, although this should be investigated more carefully.

\subsection{Is the Renormalized Energy Approach as good as conventional PDE solvers?}
\label{sec:computationalefficiency}
We take the hexagonal domain as an example to compare the computational costs of constructing the rLdG solution landscape using the renormalized energy (RE) approach versus the direct high-index optimization-based saddle dynamics (HiOSD)\cite{yin2019high} search in the rLdG model (based on finite-element/finite-difference numerical computations of solutions of the Euler-Lagrange equations in \eqref{eq:rLdG_EL}).

\paragraph{Comparison of Methodology} 
In \cite{han2021solution}, the authors use the \emph{direct approach} based on the HiOSD algorithm to explore the infinite-dimensional energy landscape of the rLdG functional in \eqref{eq:p_energy}, as a function of the re-scaled edge length $\lambda$. In contrast, the \emph{RE framework} decomposes this task as follows:\begin{enumerate}\item \textbf{RE Search:} For each 
admissible set of $(\mathbf{d}, \mathbf{c})$ configurations of degrees of interior defects and vertices, we define the corresponding renormalized energy $W$. HiOSD is then applied to $W$ to compute the critical points (representing defect positions) of $W$ in a significantly reduced search space of dimension at most $2N$, along with their Morse indices.\item \textbf{rLdG Refinement:} Based on the identified defect positions, a reference state $u_*$ is constructed from \eqref{eq:main-canonical} and used as an initial guess for a Newton-type solver of the PDEs in \eqref{eq:rLdG_EL}, yielding the corresponding full rLdG solution (when it exists).\end{enumerate}

\textbf{Benchmark (Direct Approach or Direct rLdG Search)}:
As reported in \cite{han2021solution}, the construction of the rLdG solution landscape on a hexagon domain with tangential boundary conditions, at $\lambda^2 = 150$ with a regular triangle mesh $h = 1/50$, requires a total CPU time of $46$ minutes using HiOSD on a Lenovo T450s laptop (Matlab 2018a). The average time to find a single new connection between critical points is approximately $\mathbf{6.5}$ minutes, which is almost the time required to find a new solution. The efficiency of this high-dimensional search is inherently limited by the Morse index of the target solutions and the conditioning of the Hessian.

\textbf{Efficiency of the RE-based Approach}: Our computations are conducted on a MacBook Pro (2.3 GHz 8-core Intel Core i9, 16 logical cores) using Python 3 for a system with $K=6$ and $N_{max}=4$. We firstly perform a parallelized scan across 94 sets of mixed $(\mathbf{d}, \mathbf{c})$ configurations, involving 14 unique cyclic isomorphic classes. We apply HiOSD to the $2N$-dimensional RE; the total wall-clock time to locate all critical points of $W$ is only $4.8$ minutes. We use the $38$ critical points of $W$ (computed at the RE stage) to construct high-quality initial conditions for the rLdG solver as outlined above; the Newton-based refinement yields $15$ full rLdG solutions in just $1.57$ minutes. The average time for finding a new solution is $\mathbf{0.42}$ minutes.

\begin{table}[htbp]
\centering
\caption{Computational Efficiency Comparison}
\begin{tabular}{lccc}
\hline
Method & Number of rLdG Solutions & Total Time & Average Time \\
\hline
\textbf{RE-based approach} & 15 & $4.8 + 1.57$ min & \textbf{0.42} min \\
Direct rLdG Search & 7 & 46 min & 6.5 min \\
\hline
\end{tabular}
\end{table}

The primary advantage of the RE-based method lies in its drastic dimensionality reduction, which facilitates a transition from searching in an infinite-dimensional discretized functional space to a low-dimensional manifold of defect positions (dictated by the critical points of the renormalized energy). While HiOSD remains the core engine for finding saddle points, its application to this $2N$-dimensional manifold is exponentially faster than in the original space for several fundamental reasons. Firstly, this strategy effectively mitigates the curse of dimensionality; by collapsing the search volume from a high-dimensional grid (typically $10^4 \sim 10^6$ degrees of freedom) to a space of dimension $2N$, the probability of locating isolated solution branches increases by orders of magnitude. Secondly, the computational overhead per iteration is drastically reduced, as the high-index saddle point search is performed on a small $2N \times 2N$ Hessian matrix, bypassing the expensive spectral estimations and poor conditioning inherent in discretized PDEs. Thirdly, this approach effectively bypasses the sensitivity of HiOSD to initial distances from target states or sensitivity to initial conditions. By performing a comprehensive global scan of the smooth and well-behaved renormalized energy landscape, we can identify near-perfect initial guesses ($u_*$ or $\P_*$) for the full rLdG model. This "global-to-local" pipeline ensures, to a certain extent, that the final stage of the computation—the rLdG refinement—takes place within the basin of quadratic convergence of a Newton-type solver. Moreover, the large number of initial guesses can be further parallelized for computation. Consequently, the RE-based framework provides a robust and high-speed gateway to construct complex solution landscapes that would otherwise be computationally prohibitive.  

Having said that, one should be aware of the limitations of the RE method. Firstly, the RE method yields excellent initial conditions for the rLdG model in the $\lambda \to \infty$ or the $\eps \to 0$ limit, and not for modest domain sizes. Secondly, the RE method is based on the underpinning assumption that the canonical harmonic map has isolated interior point defects and consequently, cannot capture line defects and cannot act as a search engine for rLdG critical points with high-dimensional defects (such as the WORS or BD solutions discussed in \cite{canevari2017order};
the WORS has two defect or nodal lines (with $\P \approx 0$) along the two square diagonals and the BD solution has two defect lines localised near the shortest rectangle edges). Thirdly, the RE approach cannot be readily extended to three-dimensional scenarios whereas the HiOSD has been successfully applied to study solution landscapes in the LdG framework, in cuboids and prisms 
\cite{shi2023hierarchies,shi2024multistability,han2023reduced}.

\section{Conclusion}
\label{sec:conclusion}
The rLdG energy in \eqref{eq:p_energy} is a canonical free energy in variational theories for materials science, wherein the material phase is modelled by an order parameter, $u:\mathbb{R}^2 \to \mathbb{R}^m$ with $m\geq 2$ and the energy density comprises a convex elastic energy density and a multi-well bulk potential. The bulk potential attains its minima at $|u| \in \left\{U_1, \ldots, U_N \right\}$. Such free energies appear in the context of liquid crystals, theory of superconductivity, micromagnetics, phase-field models etc. These energies are nonlinear and non-convex and typically admit multiple minima and saddle points; the multiplicity and properties of the critical points depend sensitively on the geometry, curvature, boundary conditions and model parameters. For example, it is known that the rLdG energy has a unique critical point/minimizer on small domains with Dirichlet boundary conditions and multiple minima appear as the domain size increases. It is an open question as to how the multiple solution branches disappear or annihilate as the domain size shrinks, or what is the role of defects in multistability and the relationship between the Morse index of a critical point and the nature of its defect set. These questions are of fundamental mathematical interest, i.e., how do we control or understand the solution set of systems of nonlinear and coupled partial differential equations (in this case, the Euler-Lagrange equations associated with a free energy) and of far-reaching practical importance since these questions hold clues about controllability and tunability of material systems.

The renormalized energy is a beautiful finite-dimensional mathematical concept that has direct relevance to the study of infinite-dimensional variational problems, in this case the rLdG model for 2D liquid crystal systems. The renormalized energy effectively yields candidate defect locations for a prescribed set of topological degrees of the interior and boundary defects/vortices. These defect locations yield canonical harmonic maps (which need not be solutions of the nonlinear PDEs in \eqref{eq:rLdG_EL}) but the canonical harmonic maps tend to be excellent guesses for critical points of the rLdG energy on very large domains, subject to arbitrary Dirichlet data. We generalize the concept of the renormalized energy to polygons and there are new contributions stemming from the sharp corners and the interaction between the interior defects and the corner defects. Our main contribution is that we illustrate that the renormalized energy serves as an excellent tool for locating the critical points of the rLdG energy on typical polygons, such as hexagons and rectangles. In fact, we recover known results on rLdG solution landscapes on hexagons and rectangles by using the renormalized energy to generate initial conditions for the rLdG solver, but in a significantly more computationally efficient manner. Surprisingly, we also locate new solution branches not unravelled by previous direct approaches, since the renormalized energy approach is better equipped to locate disconnected solution branches. Our work can give valuable insight into how interior point defects/vortices (where the order parameter vanishes) and vertex defect profiles (the $\mathbf{c}$) can be tuned to generate and perhaps, even stabilize solutions with tailor-made morphologies on both regular (the hexagon example in Section~\ref{sec:hexagon}) and irregular polygons (the rectangle example in Section~\ref{sec:rectangle}). This is particularly interesting since recent work shows that both polygon size and polygon asymmetry play a crucial role in multistability, in the context of the rLdG model \cite{rajamanickam2026nematic}.

Moving forward, we see several ways in which the renormalized energy could complement and supplement the study of rLdG solution landscapes - the study of critical points of the rLdG energy, their Morse indices and pathways between different critical points of the rLdG energy. For example, we have used the renormalized energy to locate several solution branches of \eqref{eq:rLdG_EL} for large $\lambda$. These solution branches are typically distinguished by arrays of interior defects. These solution branches will typically disappear as we decrease $\lambda$. The renormalized energy can offer explicit insight into the instabilities of these solution branches, i.e., are the instabilities related to the merger of defects so that higher-index saddle points connect to lower-index saddle points along pathways such that interior defects merge with each other and/or interior defects get pinned to polygon vertices along these pathways. Hence, the renormalized energy can offer insights into the passage from multistability to uniqueness in the solution landscapes of variational problems, as a function of the geometry. Further, the renormalized energy is inherently restricted to two-dimensional domains. However, we can use the renormalized energy to construct z-invariant critical points of the rLdG energy on three-dimensional domains, such as cuboids and cylinders. Again, the renormalized approach can be used to study the stability properties of such z-invariant solution branches, and in particular, the bifurcation from z-invariant to z-dependent solutions in 3D scenarios. Of course, it remains an open question as to why certain configurations - labelled by ($\mathbf{d}$, $\mathbf{c}$) - correspond to critical points of the rLdG energy and why other configurations do not, and a concrete answer to this question would give algorithmic recipes for constructing critical points of the rLdG energy with interior point defects. We hope to address some of these questions in future work.

\section{Acknowledgements} YH is supported by Young Elite Scientists Sponsorship Program by Beijing Association for Science and Technology and the Young Faculty Development Fund of the School of Mathematics at Renmin University of China. YH and AM acknowledge support from the Leverhulme Research Project Grant RPG-2021-401. 
YH thanks to Mr. Yuyang Liu from Peking University for his guidance on using the Saddlelandscape package. 



\textbf{Appendix A}

\begin{proposition}\label{prop:71}For $z \in \partial D$ on the boundary of a unit disk, and $a_j \in D$ in the unit disk, we have $$\frac{\partial}{\partial \nu} \left( \ln|1 - z \bar{a}_j| + \ln|z - a_j| \right) = 1.$$
\end{proposition}
\begin{proof}
We consider the function:$$u(z) = \ln|1 - z \bar{a}_j| + \ln|z - a_j|.$$
Let $z = re^{i\theta}$. To compute the radial derivative, we express the terms in $u(z)$ using the identity $\ln|w| = \frac{1}{2}\ln(w\bar{w})$,
\begin{align}
u(re^{i\theta}) &= \frac{1}{2} \ln \left( (1 - r e^{i\theta} \bar{a}_j)(1 - r e^{-i\theta} a_j) \right) + \frac{1}{2} \ln \left( (r e^{i\theta} - a_j)(r e^{-i\theta} - \bar{a}_j) \right)\\
&= \frac{1}{2} \ln \left( 1 - r(e^{i\theta} \bar{a}_j + e^{-i\theta} a_j) + r^2 |a_j|^2 \right) + \frac{1}{2} \ln \left( r^2 - r(e^{i\theta} \bar{a}_j + e^{-i\theta} a_j) + |a_j|^2 \right).
\end{align}
On the boundary $\partial D$, the outward normal derivative $\frac{\partial}{\partial \nu}$ corresponds to the radial derivative $\frac{\partial}{\partial r}$ in polar coordinates $z = re^{i\theta}$ evaluated at $r=1$.
Differentiating with respect to $r$ yields $$\frac{\partial u}{\partial r} = \frac{1}{2} \left[ \frac{-(e^{i\theta} \bar{a}_j + e^{-i\theta} a_j) + 2r |a_j|^2}{1 - r(e^{i\theta} \bar{a}_j + e^{-i\theta} a_j) + r^2 |a_j|^2} + \frac{2r - (e^{i\theta} \bar{a}_j + e^{-i\theta} a_j)}{r^2 - r(e^{i\theta} \bar{a}_j + e^{-i\theta} a_j) + |a_j|^2} \right].$$Setting $r=1$ for the boundary $\partial D$,
we have
$$\left. \frac{\partial u}{\partial r} \right|_{r=1} = \frac{1}{2} \left[ \frac{2|a_j|^2 - 2(e^{i\theta} \bar{a}_j + e^{-i\theta} a_j) + 2 }{1 - (e^{i\theta} \bar{a}_j + e^{-i\theta} a_j) + |a_j|^2} \right]= 1.\square$$
\end{proof}
\textbf{Appendix B}
\begin{proposition}\label{prop:tildeR0}
The Laplace Neumann boundary problem is
\begin{align}
&\Delta \tilde{R}_0 = 0,\ \text{on}\ D\\
&\frac{\partial \tilde{R}_0}{\partial \mu} =  \sum_{j = 1}^K \frac{2\pi}{\alpha_j} c_j \frac{\partial \ln|f(w)-b_j|}{\partial \mu} - \frac{\partial \sum_{i = 1}^N d_i \ln|f(w)-a_i|}{\partial \mu}\ \text{on}\ \partial D\backslash\{w_1^{(c)},\ldots,w_K^{(c)}\}.
\end{align}
where $\mu$ is the outward normal direction on the boundary of the disk, $a_i = f(w_i)$, $i = 1,\ldots,N$ $b_j = f(w_j^{(c)})$, $j = 1,\ldots, K$. Then we have
\begin{equation}\label{eq:tildeR0}
\tilde{R}_0(w) = \sum_{j=1}^K \frac{2\pi}{\alpha_j} c_j \ln \left| \frac{f(w) - b_j}{(w - w_j^{(c)})^{\alpha_j/\pi}} \right| + \sum_{i=1}^N d_i \ln \left| \frac{(1 - w \bar{w}_i)(w - w_i)}{f(w) - a_i} \right|.
\end{equation}
\end{proposition}
\begin{proof}
Due to the linearity of the Laplace operator, we can decompose the regularized field into two harmonic components,$$\tilde{R}_0(w) = \tilde{R}_{0,\text{defect}}(w) + \tilde{R}_{0,\text{corner}}(w)$$
where each component handles the flux generated by the interior defects and the polygonal corners, respectively, while jointly satisfying the Neumann compatibility condition.
The defect component $\tilde{R}_{0,\text{defect}}$ satisfies
\begin{equation}\begin{cases}
\Delta \tilde{R}_{0,\text{defect}} = 0, \quad & w \in D \\
\frac{\partial \tilde{R}_{0,\text{defect}}}{\partial \mu} = - \frac{\partial \sum_{i=1}^N d_i \ln|f(w) - a_i|}{\partial \mu}  + \sum_{i = 1}^N d_i, \quad & w \in \partial D\label{eq:defect_bc}\end{cases}\end{equation}
In order for the Neumann problem to have a solution, the Neumann boundary condition must satisfy\begin{equation}\label{eq:Neumann}\int_{\partial D} \frac{\partial\tilde{R}_{0,\text{defect}}}{\partial \mu} dS = 0,
\end{equation}
which is the origin of the term $\sum_{i = 1}^N d_i$. We first decompose the singular potential generated by the mapped defects
\begin{equation}\label{eq:fa}
\ln|f(w) - a_i| = \ln|w - w_i| + \ln\left| \frac{f(w) - a_i}{w - w_i} \right|.
\end{equation}
As $w\to w_i$, the limit of the second term in \eqref{eq:fa} is $\ln|f'(w_i)|$. According to \eqref{eq:f'}, $f'$ has no singularity inside $D$.
To satisfy the boundary condition in \eqref{eq:defect_bc} for the first term, we use the following equation proved in Proposition \ref{prop:71} $$\frac{\partial}{\partial \mu} \left( \ln|1 - w\bar{w}_i| + \ln|w - w_i| \right) = 1 \quad \text{on } \partial D.$$
Thus, the analytical solution for the defect component is
\begin{equation}\label{eq:R0defect}
\tilde{R}_{0,\text{defect}}(w) = \sum_{i=1}^N d_i \ln|1 - w\bar{w}_i| - \sum_{i=1}^N d_i \ln\left| \frac{f(w) - a_i}{w - w_i} \right|
\end{equation}
The corner component: $\tilde{R}_{0,\text{corner}}$ satisfies $$\begin{cases}
\Delta \tilde{R}_{0,\text{corner}} = 0, \quad & w \in D \\
\frac{\partial \tilde{R}_{0,\text{corner}}}{\partial \mu} = \sum_{j=1}^K \frac{2\pi}{\alpha_j} c_j \frac{\partial \ln|f(w) - b_j|}{\partial \mu}  - \sum_{i=1}^N d_i, \quad & w \in \partial D \setminus \{w_1^{(c)}, \ldots, w_K^{(c)}\}
\end{cases}$$
The term $\ln|f(w) - b_j|$ is harmonic inside $D$. To absorb the residual flux $- \sum d_i$ generated by the defect component, we introduce a compensating harmonic field centered at the boundary pre-vertices $w_j^{(c)}$.
We utilize the topological charge conservation condition, $\sum_{j=1}^K c_j = \sum_{i=1}^N d_i$. A fundamental property of the logarithmic potential evaluated on the boundary of a unit disk is
$$\frac{\partial}{\partial \mu} \ln|w - w_j^{(c)}| = \frac{1}{2} \quad \text{on } \partial D \setminus\{w_1^{(c)}, \ldots, w_K^{(c)}\}.$$
Subtracting the term $2 \sum_{j=1}^K c_j \ln|w - w_j^{(c)}|$ perfectly neutralizes the residual flux $- \sum d_i$. Multiplying by the geometric factor, the analytical solution for the corner component is
\begin{equation}\label{eq:R0corner}
\tilde{R}_{0,\text{corner}}(w) = \sum_{j=1}^K \frac{2\pi}{\alpha_j} c_j \ln \left| \frac{f(w) - b_j}{(w - w_j^{(c)})^{\alpha_j/\pi}} \right|
\end{equation}
Remarkably, according to the asymptotic behavior of the general Schwarz-Christoffel mapping near the vertices ($f(w) - b_j \sim (w - w_j^{(c)})^{\alpha_j/\pi}$), this specific combination entirely eliminates the boundary singularities, rendering $\tilde{R}_{0,\text{corner}}$ smooth and well-defined even at the pre-vertices $w = w_j^{(c)}$.
Combining both components in \eqref{eq:R0defect} and \eqref{eq:R0corner}, we obtain the analytical solution for the regularized field $\tilde{R}_0(w)$ on the entire unit disk in \eqref{eq:tildeR0}.
\end{proof}
\bibliographystyle{unsrt}
\bibliography{main}

@article{han2022elastic,
  title={Elastic anisotropy in the reduced {Landau--de Gennes} model},
  author={Han, Yucen and Harris, Joseph and Majumdar, Apala and Zhang, Lei},
  journal={Proceedings. Mathematical, Physical, and Engineering Sciences},
  volume={478},
  number={2261},
  pages={20210966},
  year={2022}
}

@book{de1993physics,
  title={The physics of liquid crystals},
  author={De Gennes, Pierre-Gilles and Prost, Jacques},
  number={83},
  year={1993},
  publisher={Oxford {U}niversity {P}ress},
  edition = {2},
  address={Oxford}
}

@article{lagerwall2012new,
  title={A new era for liquid crystal research: Applications of liquid crystals in soft matter nano-, bio-and microtechnology},
  author={Lagerwall, Jan PF and Scalia, Giusy},
  journal={Current Applied Physics},
  volume={12},
  number={6},
  pages={1387--1412},
  year={2012},
  publisher={Elsevier}
}

@book{wojtowicz1975introduction,
  title={Introduction to liquid crystals},
  author={Wojtowicz, Peter J and Sheng, Ping and Priestley, EB},
  year={1975},
  publisher={Springer}
}

@book{bethuel1994ginzburg,
  title={{Ginzburg-Landau} vortices},
  author={Bethuel, Fabrice and Brezis, Ha{\"\i}m and H{\'e}lein, Fr{\'e}d{\'e}ric and others},
  volume={13},
  year={1994},
  publisher={Springer}
}

@inproceedings{lin1995solutions,
  title={Solutions of {Ginzburg-Landau} equations and critical points of the renormalized energy},
  author={Lin, Fang Hua},
  booktitle={Annales de l'Institut Henri Poincar{\'e} C, Analyse non lin{\'e}aire},
  volume={12},
  number={5},
  pages={599--622},
  year={1995},
  organization={Elsevier}
}

@article{majumdar2004elastic,
  title={Elastic energy of liquid crystals in convex polyhedra},
  author={Majumdar, Apala and Robbins, JM and Zyskin, M},
  journal={Journal of Physics A: Mathematical and General},
  volume={37},
  number={44},
  pages={L573},
  year={2004},
  publisher={IOP Publishing}
}

@book{cormen2022introduction,
  title={Introduction to algorithms},
  author={Cormen, Thomas H and Leiserson, Charles E and Rivest, Ronald L and Stein, Clifford},
  year={2022},
  publisher={MIT press}
}

@article{bethuel1993asymptotics,
  title={Asymptotics for the minimization of a {Ginzburg-Landau} functional},
  author={Bethuel, Fabrice and Brezis, Ha{\"\i}m and H{\'e}lein, Fr{\'e}d{\'e}ric},
  journal={Calculus of Variations and Partial Differential Equations},
  volume={1},
  number={2},
  pages={123--148},
  year={1993},
  publisher={Springer}
}

@article{yin2019high,
  title={High-index optimization-based shrinking dimer method for finding high-index saddle points},
  author={Yin, Jianyuan and Zhang, Lei and Zhang, Pingwen},
  journal={SIAM Journal on Scientific Computing},
  volume={41},
  number={6},
  pages={A3576--A3595},
  year={2019},
  publisher={SIAM}
}

@article{yin2020construction,
  title={Construction of a pathway map on a complicated energy landscape},
  author={Yin, Jianyuan and Wang, Yiwei and Chen, Jeff ZY and Zhang, Pingwen and Zhang, Lei},
  journal={Physical Review Letters},
  volume={124},
  number={9},
  pages={090601},
  year={2020},
  publisher={APS}
}

@article{han2021solution,
  title={Solution landscape of a reduced {Landau--de Gennes} model on a hexagon},
  author={Han, Yucen and Yin, Jianyuan and Zhang, Pingwen and Majumdar, Apala and Zhang, Lei},
  journal={Nonlinearity},
  volume={34},
  number={4},
  pages={2048--2069},
  year={2021},
  publisher={IOP Publishing}
}

@book{Logg2012FEniCS,
  author    = {Logg, Anders and Mardal, Kent-Andre and Wells, Garth},
  title     = {Automated Solution of Differential Equations by the Finite Element Method: The {FEniCS Book}},
  series    = {Lecture Notes in Computational Science and Engineering},
  volume    = {84},
  publisher = {Springer},
  address   = {Berlin},
  year      = {2012}
}

@misc{liu2026saddlescapev10pythonpackage,
      title={{SaddleScape V1.0: A Python Package for Constructing Solution Landscapes via High-index Saddle Dynamics}}, 
      author={Yuyang Liu and Hua Su and Zixiang Xiao and Lei Zhang and Jin Zhao},
      year={2026},
      eprint={2601.01081},
      archivePrefix={arXiv},
      primaryClass={math.NA},
      url={https://arxiv.org/abs/2601.01081}, 
}

@article{shi2022nematic,
  author  = {Shi, Baoming and Han, Yucen and Zhang, Lei},
  title   = {Nematic liquid crystals in a rectangular confinement: Solution landscape, and bifurcation},
  journal = {SIAM Journal on Applied Mathematics},
  volume  = {82},
  number  = {5},
  pages   = {1808--1828},
  year    = {2022}
}

@book{HornJohnson2012,
  author    = {Horn, Roger A. and Johnson, Charles R.},
  title     = {Matrix Analysis},
  edition   = {2},
  publisher = {Cambridge University Press},
  address   = {Cambridge},
  year      = {2012},
  doi       = {10.1017/CBO9781139020411},
  isbn      = {9780521839402}
}

@article{Oppenheim1930,
  title={Inequalities connected with definite Hermitian forms},
  author={Oppenheim, Alexander},
  journal={Journal of the London Mathematical Society},
  volume={1},
  number={2},
  pages={114--119},
  year={1930},
  publisher={Wiley Online Library}
}

@article{Schechter1959,
  author  = {Schechter, Samuel},
  title   = {On the Inversion of Certain Matrices},
  journal = {Mathematical Tables and Other Aids to Computation},
  volume  = {13},
  number  = {66},
  pages   = {73--77},
  year    = {1959},
  doi     = {10.2307/2001955}
}

@article{andreshafrir1998,
  author  = {Andr{\'e}, Nelly and Shafrir, Itai},
  title   = {Minimization of a {Ginzburg--Landau} Type Functional with Nonvanishing Dirichlet Boundary Condition},
  journal = {Calculus of Variations and Partial Differential Equations},
  volume  = {7},
  number  = {3},
  pages   = {191--217},
  year    = {1998},
  doi     = {10.1007/s005260050106},
  url     = {https://doi.org/10.1007/s005260050106}
}

@article{alamabronsard2023,
  author  = {Alama, Stan and Bronsard, Lia and van Brussel, Lee},
  title   = {On Minimizers of the 2D {Ginzburg--Landau} Energy with Tangential Anchoring},
  journal = {Nonlinear Analysis},
  volume  = {232},
  pages   = {113276},
  year    = {2023},
  doi     = {10.1016/j.na.2023.113276},
  url     = {https://doi.org/10.1016/j.na.2023.113276}
}

@article{alamabronsard2024,
  author  = {Alama, Stan and Bronsard, Lia and Colinet, Andrew},
  title   = {$\Gamma$-Convergence of the {Ginzburg--Landau} Functional with Tangential Boundary Conditions},
  journal = {Journal of Functional Analysis},
  volume  = {287},
  number  = {11},
  pages   = {110621},
  year    = {2024},
  doi     = {10.1016/j.jfa.2024.110621},
  url     = {https://doi.org/10.1016/j.jfa.2024.110621}
}

@article{shi2023hierarchies,
  title={Hierarchies of critical points of a {Landau-de Gennes} free energy on three-dimensional cuboids},
  author={Shi, Baoming and Han, Yucen and Yin, Jianyuan and Majumdar, Apala and Zhang, Lei},
  journal={Nonlinearity},
  volume={36},
  number={5},
  pages={2631--2654},
  year={2023},
  publisher={IOP Publishing}
}

@article{shi2024multistability,
  title={Multistability for nematic liquid crystals in cuboids with degenerate planar boundary conditions},
  author={Shi, Baoming and Han, Yucen and Majumdar, Apala and Zhang, Lei},
  journal={SIAM Journal on Applied Mathematics},
  volume={84},
  number={2},
  pages={756--781},
  year={2024},
  publisher={SIAM}
}

@article{han2023reduced,
  title={A reduced {Landau-de Gennes} study for nematic equilibria in three-dimensional prisms},
  author={Han, Yucen and Shi, Baoming and Zhang, Lei and Majumdar, Apala},
  journal={IMA Journal of Applied Mathematics},
  volume={88},
  number={5},
  pages={645--676},
  year={2023},
  publisher={Oxford University Press}
}

@article{han2019transition,
  title={Transition pathways between defect patterns in confined nematic liquid crystals},
  author={Han, Yucen and Hu, Yucheng and Zhang, Pingwen and Zhang, Lei},
  journal={Journal of Computational Physics},
  volume={396},
  pages={1--11},
  year={2019},
  publisher={Elsevier}
}

@article{majumdar2010equilibrium,
  title={Equilibrium order parameters of nematic liquid crystals in the {Landau-de Gennes} theory},
  author={Majumdar, Apala},
  journal={European Journal of Applied Mathematics},
  volume={21},
  number={2},
  pages={181--203},
  year={2010},
  publisher={Cambridge University Press}
}

@article{patranabish2021quantum,
  title={Quantum-dots-dispersed bent-core nematic liquid crystal and cybotactic clusters: experimental and theoretical insights},
  author={Patranabish, Sourav and Wang, Yiwei and Sinha, Aloka and Majumdar, Apala},
  journal={Physical Review E},
  volume={103},
  number={5},
  pages={052703},
  year={2021},
  publisher={APS}
}

@article{rajamanickam2026nematic,
  title={Nematic equilibria in isosceles triangles: The effects of edge length and apex angle on solution landscapes in a reduced {Landau--de Gennes} framework},
  author={Rajamanickam, Prabakaran and Han, Yucen and Hinai, Thuriya AL and Majumdar, Apala},
  journal={La Matematica},
  volume={5},
  number={3},
  pages={53},
  year={2026},
  publisher={Springer}
}

@article{canevari2017order,
  title={Order reconstruction for nematics on squares and hexagons: A {Landau--de Gennes} study},
  author={Canevari, Giacomo and Majumdar, Apala and Spicer, Amy},
  journal={SIAM Journal on Applied Mathematics},
  volume={77},
  number={1},
  pages={267--293},
  year={2017},
  publisher={SIAM}
}

@misc{han2026asymptotictheoryginzburglandaucritical,
      title={The Asymptotic Theory of Ginzburg--Landau Critical Points on Domains with Corners}, 
      author={Yucen Han and Apala Majumdar},
      year={2026},
      eprint={2610.01410},
      archivePrefix={arXiv},
      primaryClass={math.AP},
      url={https://arxiv.org/abs/2610.01410}, 
}
\end{document}